\documentclass[12pt,journal,compsoc]{IEEEtran}

\usepackage{mathrsfs}
\usepackage{hyperref}
\usepackage{courier}
\usepackage{epsfig}
\usepackage{graphicx}
\usepackage{amsmath}
\usepackage{amsthm}
\usepackage{mathtools}
\usepackage{amssymb}
\usepackage{verbatim}
\usepackage{amscd}
\usepackage{times}
\usepackage{subcaption}
\usepackage{lipsum}
\usepackage{color}
\usepackage{enumerate}
\usepackage{booktabs}
\usepackage{bm}
\usepackage{float}
\usepackage{diagbox}
\usepackage{pgfplots}

\newtheorem{thmA}{Theorem}

\newtheorem{thm}{Theorem}

\newtheorem{lem}{Lemma}
\newtheorem{prp}{Proposition}
\newtheorem{dfn}{Definition}

\newtheorem{remark}{Remark}

\def\As{\mathscr{A}}
\def\Is{\mathscr{I}}
\def\Lc{\mathcal{L}}
\def\Ls{\mathscr{L}}
\def\M{\mathcal{M}}
\def\Ps{\mathscr{P}}
\def\P{\mathcal{P}}
\def\Rs{\mathscr{R}}
\def\Ts{\mathscr{T}}
\def\Ys{\mathscr{Y}}
\def\H{\mathcal{H}}
\def\I{\mathcal{I}}
\def\J{\mathcal{J}}
\def\Re{\mathbb{R}}

\def\V{\mathcal{V}}

\def\W{\mathcal{W}}

\DeclareMathOperator*{\Ass}{Ass}
\DeclareMathOperator*{\Min}{Min}

\DeclareMathOperator*{\height}{height}

\DeclareMathOperator*{\grade}{grade}

\DeclareMathOperator*{\rank}{rank}

\DeclareMathOperator*{\In}{in}

\begin{document} 

\graphicspath{{figures/}}

\title{Ladder Matrix Recovery from Permutations}

\author{
    \IEEEauthorblockN{Manolis C. Tsakiris}\\ 
    \vspace{0.2in}
    \IEEEauthorblockA{\small Key Laboratory of Mathematics Mechanization \\ Academy of Mathematics and Systems Science \\ Chinese Academy of Sciences \\ Beijing, 100190, China \\
    manolis@amss.ac.cn}}\normalsize

\IEEEtitleabstractindextext{
\begin{abstract}
We give unique recovery guarantees for matrices of bounded rank that have undergone permutations of their entries. We even do this for a more general matrix structure that we call ladder matrices. We use methods and results of commutative algebra and algebraic geometry, for which we include a preparation as needed for a general audience. 
\end{abstract}
\begin{IEEEkeywords}
Unlabeled Sensing, Unlabeled Principal Component Analysis, Determinantal Varieties, Gr\"obner Bases
\end{IEEEkeywords}}

\maketitle

\section{Introduction}\label{section:Introduction}

In \cite{Unnikrishnan-Allerton2015,Unnikrishnan-TIT18} Unnikrishnan, Haghighatshoar \& Vetterli considered the problem of solving a linear system of equations, for which the right-hand-side vector is known only up to a permutation and it may even have missing entries. They termed this problem \emph{unlabeled sensing}. Their main result was that for a generic coefficient matrix, unique recovery of the original solution is possible, as long as at least $2r$ unordered entries of the right-hand-side vector are given, where $r$ is the dimensionality of the solution. In an algebraic-geometric theory termed \emph{homomorphic sensing} (\cite{tsakiris2018eigenspace,Tsakiris-ICML2019,Peng-ACHA-21}), Peng and the author showed that this phenomenon can be explained from an abstract point of view, for arbitrary linear transformations of the right-hand-side vector. 

Besides its theoretical interest, which has attracted several authors (e.g., \cite{Hsu-NIPS17,Pananjady-TIT18,Slawski-JoS19,zhang2020optimal,tsakiris2020algebraic,jeong2020recovering}), unlabeled sensing entertains more than a few potential applications. These include record linkage (\cite{Slawski-JoS19,Slawski-JMLR2020,Slawski-JCGS2021}), image and point cloud registration (\cite{li2021generalized}),  cell sorting (\cite{Abid-Allerton2018,xie2021hypergradient}), metagenomics (\cite{ma2021optimal}), neuron matching (\cite{nejatbakhsh2021neuron}), spatial field estimation (\cite{kumar2017bandlimited}), and target localization (\cite{wang2020target}). 

In \cite{yao2021} Yao, Peng and the author considered a natural extension of unlabeled sensing, termed \emph{unlabeled principal component analysis}, from vector spaces to  algebraic varieties of matrices of bounded rank (which as a special case includes the low-rank matrices). A data matrix of rank $r$ is observed only up to an arbitrary permutation of each column, and the question is whether it uniquely defines the original matrix, necessarily, up to a row permutation. The main result of \cite{yao2021} asserts that this is indeed the case for a generic matrix of rank $r$. 

The aim of the present article is to establish a unique recovery result in the context of unlabeled principal component analysis, stronger than that of \cite{yao2021}, with the rank-$r$ matrix having undergone an entirely arbitrary permutation of its entries, as opposed to the column-wise permutations of \cite{yao2021}. We do this by invoking powerful results in the field of commutative algebra (\cite{eisenbud2013commutative}), and specifically in the theory of Gr\"obner bases (\cite{cox2013ideals}) for determinantal ideals (\cite{sturmfels1990grobner,bruns2003grobner}). In fact, our recovery results are established for a more general matrix structure that we introduce, which we call \emph{ladder matrix}. Ladder matrices are matrices that consist of non-zero entries only inside a ladder-like region of the matrix, as shown in Fig. \ref{fig:ladder}. 
\begin{figure}[h]
\centering
\includegraphics[width=0.4\textwidth]{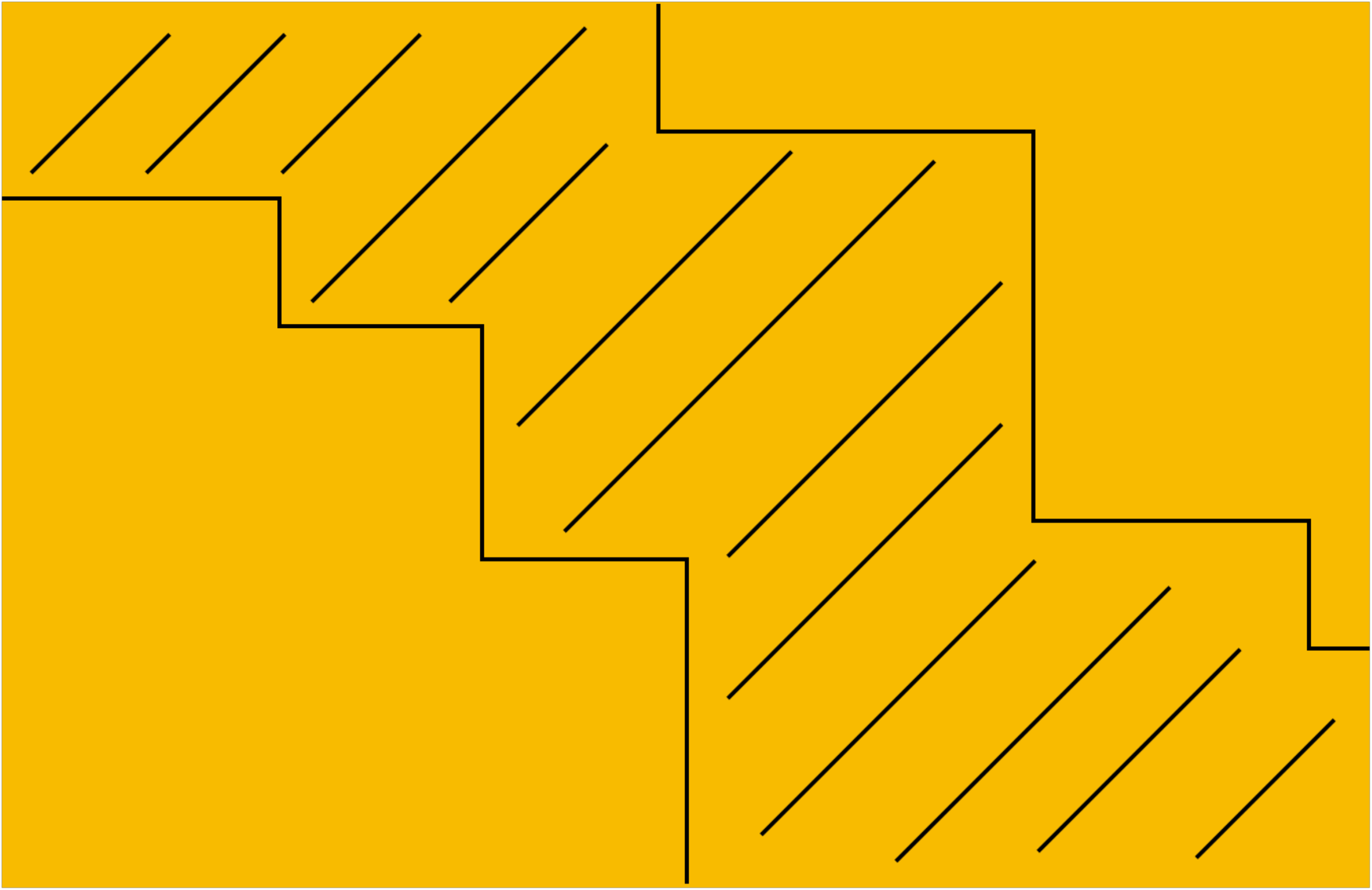}
\caption{A matrix with non-zero entries occuring only inside the dashed area. We call such a matrix a \emph{ladder matrix}.} \label{fig:ladder}	
\end{figure}
Such a data matrix structure arises naturally in neuroscience, by suitably aligning neuronal recordings based on some cue onset (e.g., see Fig. 2j in \cite{peters2014emergence}). It also occurs when one observes over time features of a dynamic process, such as the gene expression levels of different genes of an organism across its life span (e.g., see Fig. 3b in \cite{he2020changing}). Another possible source of such data is due to structured occlusions. Requiring the rank of the data to be bounded inside the ladder, leads to the well-known notion in commutative algebra of a \emph{ladder determinantal variety} (\cite{narasimhan1986irreducibility,conca1995ladder,conca1996gorenstein,conca1997ladder}), which encompasses as a particular instance the variety of matrices of bounded rank. We will call such matrices \emph{ladder matrices of bounded ladder-rank}. Then our most general result stems from the theory of Gr\"obner bases for ladder determinantal ideals \cite{narasimhan1986irreducibility,gorla2007mixed}. 

\S \ref{section:Preliminaries} reviews the problems of unlabeled sensing and unlabeled principal component analysis, from which this paper is inspired. \S \ref{section:Main-Results} presents the main results of the paper; here technicalities are kept to a minimum. \S \ref{section:Proofs} lays down the proofs. These make use of methods and results of commutative algebra and algebraic geometry, employed in an expository yet rigorous fashion, for which a preparation is offered in Appendices \ref{appendix:Ring-Theory} and \ref{appendix:GB}. 

\section{Preliminaries} \label{section:Preliminaries}

\subsection{Unlabeled Sensing}
For the sake of a gradual exposition, and to give the reader a more complete view of the context, we begin by formally reviewing unlabeled sensing \cite{Unnikrishnan-TIT18}. 

Let $A \in \Re^{m \times r}$ be a matrix of rank $r$, $ \W \subset \Re^m$ its column-space, $b \in \W$ some vector. It follows that the linear system of equations $Ax=b$ has a unique solution $\xi$. 
Let $\pi:\Re^m \rightarrow \Re^m$ be a permutation on the $m$ coordinates of $\Re^m$ and $\rho:\Re^m \rightarrow \Re^\ell$ a coordinate projection of $\Re^m$, which discards all but $\ell$ entries of a vector. Set $y = \rho \circ \pi (b)$. The problem of unlabeled sensing, in its noiseless algebraic form, consists of the following fundamental question: Under what conditions do the data $A, y$ uniquely determine the solution $\xi$ of the original system of equations? The main result of \cite{Unnikrishnan-TIT18} is:

\begin{thmA}[\cite{Unnikrishnan-TIT18}] \label{thm:US}
Suppose that $A \in \Re^{m \times r}$ is generic and that the projection $\rho:\Re^m \rightarrow \Re^\ell$ preserves at least $2r$ entries, i.e. $\ell \ge 2r$. Then $A$ and $y$ uniquely determine the solution $\xi$ of the linear system of equations $Ax=b$. 
\end{thmA}

In Theorem \ref{thm:US} the attribute \emph{generic} is meant in the precise sense of the \emph{Zariski topology} of $\Re^{m \times r}$: the statement of the theorem is true for every matrix $A$ whose entries do not satisfy certain (non-zero) polynomial equations in $m r$ variables. For a concise and accessible discussion of the Zariski topology we refer the reader to \S 2.1.1 in \cite{tsakiris2020algebraic} and \S 3.1 in \cite{tsakiris2020low}, while for a systematic and still accessible treatment we recommend \cite{cox2013ideals} and \cite{michalek2021invitation}.   

\subsection{Unlabeled PCA} 

Now consider a data matrix $X =[x_1,\dots,x_n]\in \Re^{m \times n}$ of rank $r < \min\{m,n\}$, with $x_j$ the $j$th column. Let $\P_m$ be the group of permutations of the $m$ coordinates of $\Re^m$ and let $\P_m^n = \prod_{j=1}^n \P_m$ be the cartesian product of $n$ copies of $\P_m$. Thus a $\underline{\pi} \in \P_m^n$ is an ordered collection $\underline{\pi}=(\pi_1,\dots,\pi_n)$ of $n$ permutations of $\Re^m$. The set $\P_m^n$ is a group as well. It induces an action on $\Re^{m \times n}$ by sending $X$ and $\underline{\pi}$ to $Y=[y_1, \cdots, y_n]:=\underline{\pi}(X) = [\pi_1(x_1), \cdots, \pi_n(x_n)]$. The problem that was posed in \cite{yao2021} is under what conditions the data $Y$ uniquely determine $X$, necessarily up to a row permutation. \cite{yao2021} established:

\begin{thmA}[\cite{yao2021}] \label{thm:UPCA}
Let $X$ be a generic matrix of rank $r$, $\underline{\pi} \in \P_m^n$ and $Y=\underline{\pi}(X)$. 
Then for any $\underline{\pi}' \in \P_m^n$, the rank of $\underline{\pi}' (Y)$ is greater than $r$, unless $\underline{\pi}' = (\sigma,\dots,\sigma) \underline{\pi}^{-1}(Y)$, for any $\sigma \in \P_m$.\end{thmA}

In Theorem \ref{thm:UPCA} the attribute \emph{generic} is again meant in the sense of the Zariski topology, this time on the algebraic variety $$\M_{r, m \times n} = \{W \in \Re^{m \times n} : \operatorname{rank}(W) \le r \} $$ of matrices of rank at most $r$ (see the discussion in \S 2.2 in \cite{yao2021} and in \S 3.2 in \cite{tsakiris2020low}). Thus the statement of Theorem \ref{thm:UPCA} holds for every matrix of rank at most $r$, whose entries do not satisfy some non-trivial polynomial equations in $mn$ variables. By \emph{non-trivial}, we mean that these equations should not be identically satisfied by every matrix in $\M_{r, m \times n}$. (Contrast this to unlabeled sensing, where the matrix space of interest is the vector space $\Re^{m \times r}$, and a non-zero polynomial already defines a subset of $\Re^{m \times r}$ of measure zero.)

There is an interesting connection of unlabeled sensing with unlabeled principal component analysis. In unlabeled sensing, we have a linear subspace $\W$ {\textemdash the column-space of $A \in \Re^{m \times r}$\textemdash} and a permuted version $y$ of a point $b \in \W$ inside that subspace, and we ask for unique recovery of $b$ from $A$ and $y$. Instead, in unlabeled principal component analysis we have a linear subspace $\W$ {\textemdash the column space of the data matrix $X$\textemdash} and $n$ permutations $y_1,\dots,y_n$ of $n$ points $x_1,\dots,x_n$ in $\W$. Had we known $\W$, the problem would reduce to $n$ instances of an unlabeled sensing problem. However, access to $\W$ is not available, and its recovery is implicitly required (up to a permutation of $\Re^m$, as noted above). This relationship is exploited in the proof of \cite{yao2021}.
Finally, there is an analogy of this picture to low-rank matrix completion which is worth to describe: completing a sufficiently observed vector from a sufficiently generic known linear subspace, amounts to solving a linear system of equations. Completing a partially observed version of $X$ would thus consist of solving $n$ linear systems of equations, providing the column-space $\W$ of the matrix were known. Alas, this knowledge is not available, and low-rank matrix completion is as hard as computing $\W$ from the coordinate projections of $n$ points in $\W$ (the columns of the partially observed matrix); this point of view is adopted in \S 4 of \cite{tsakiris2020low}. The above discussion is summarized in Fig. \ref{fig:US-UPCA}.	

\begin{figure}[h]
\centering
\includegraphics[width=0.4\textwidth]{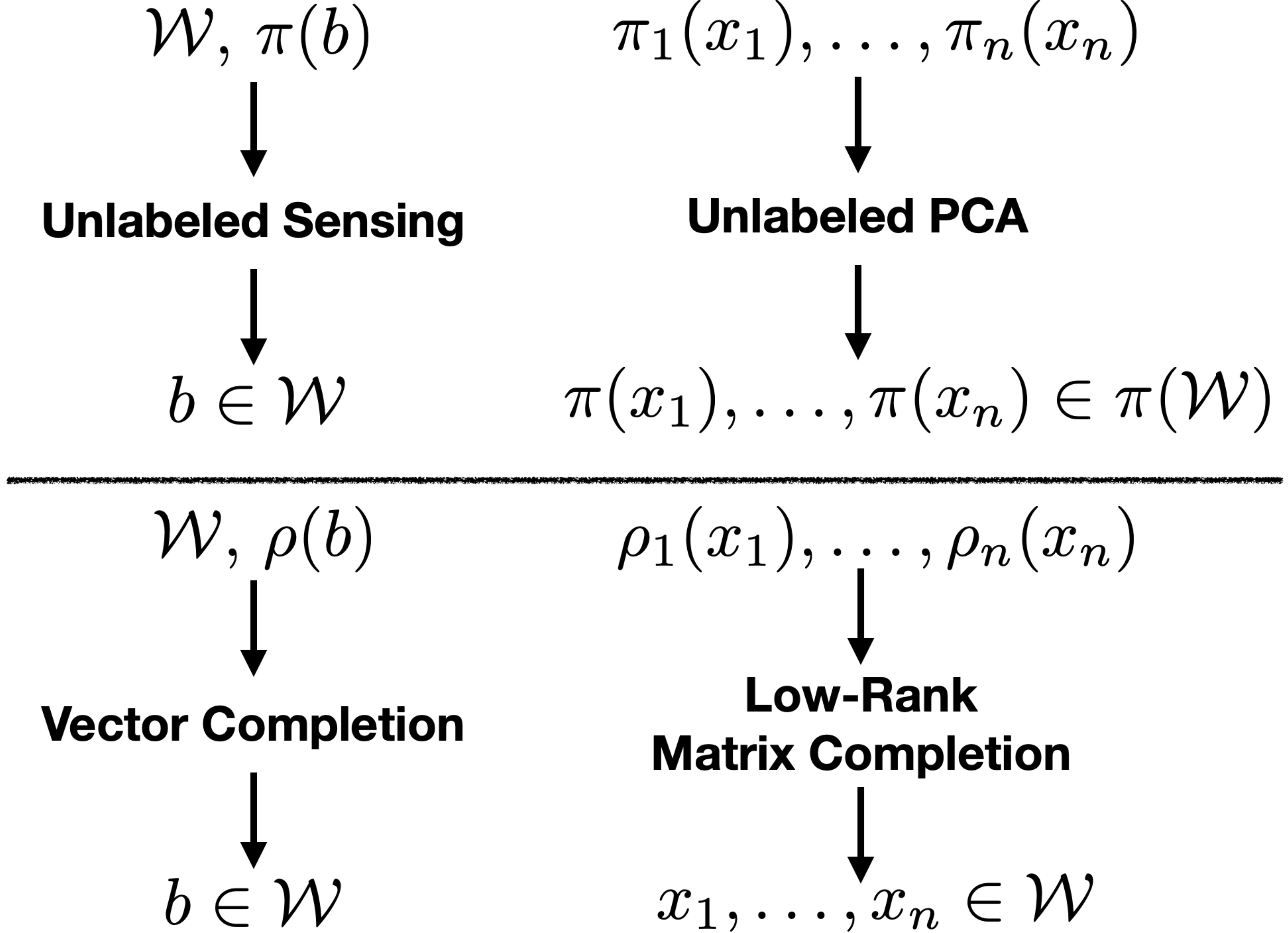}
\caption{A comparison of problems. The $\pi$'s are permutations and the $\rho$'s are coordinate projections.} \label{fig:US-UPCA}	
\end{figure}

\section{Recovery From Permutations} \label{section:Main-Results}

\subsection{Matrices of Bounded Rank}

\subsubsection{Uniqueness}

Let $\P_{m \times n}$ be the group of permutations of the coordinates of $\Re^{m \times n}$ of size $(mn)!$, and let $X \in \M_{r, m \times n}$ be our data matrix of rank at most $r$. For $\pi \in \P_{m \times n}$ we let $\pi(X)$ be the $m \times n$ matrix obtained from $X$ by permuting its entries according to $\pi$. Compared to $\rank(X)$, the rank of $\pi(X)$ can increase, it can stay the same or it can drop, depending on $X$ and $\pi$. There is an obvious family of permutations $\pi \in \P_{m \times n}$ which preserve the rank: those that only permute rows and columns of $X$. Our main result asserts that these are essentially the only rank-preserving permutations, except when $X$ belongs to a pathological zero-measure set of $\M_{r, m \times n}$. 

\begin{thm} \label{thm:arbitrary-pi}
If $X$ is a generic matrix of rank $r<\min\{m,n\}$, then $\rank[\pi(X)]>r$, except when $\pi(X) = \Pi_1 X \Pi_2$ for permutation matrices $\Pi_1$ and $\Pi_2$ or when $\pi(X) = X^\top$ in case $m = n$. 
\end{thm}

Theorem \ref{thm:arbitrary-pi} implies that when $X$ is generic, for any $\pi \in \P_{m \times n}$, the matrix $\pi(X)$ uniquely defines $X$ up to the least possible ambiguity of a row and column permutation. Indeed, one may try each and every permutation $\pi' \in \P_{m \times n}$ on $\pi(X)$, until a $\pi'$ is found for which the rank of $\pi' \pi(X)$ becomes $r$. The theorem guarantees that when this happens, $\pi' \pi(X)$ is a row and column permutation of $X$. Typically, say, the columns of $X$ correspond to samples of a population and the rows correspond to features that one measures on that population. Reordering the rows and the columns of the data matrix $X$ can be harmless, depending on the nature of the application. 

\subsubsection{A Polynomial System of Equations}

Theorem \ref{thm:arbitrary-pi} serves as a guarantee that the brute force approach described above will produce $X$ from $\pi(X)$ up to an ambiguity of row and column permutations, in analogy to Theorem \ref{thm:US} for unlabeled sensing. In unlabeled sensing, the polynomial system solving approach of  \cite{tsakiris2020algebraic} has been shown to be more efficient than brute force (see also \cite{melanova2022recovery}). In this section, we set the theoretical foundations of a similar approach for the problem at hand. 

For a positive integer $a$ we let $[a] = \{1,\dots,a\}$. Let $Z = (z_{ij})$ be an $m \times n$ matrix of variables and let $\Re[Z] = \Re\big[z_{ij}: \, i \in [m], \, j \in [n]\big]$ be the set of polynomials in variables $z_{ij}$ with real coefficients. For any positive integer $\nu$ define a polynomial $$ p_\nu(Z) = \sum_{i \in [m], \, j\in [n]} z_{ij}^\nu.$$ Note that $p_\nu(Z)$ is symmetric, that is, for any permutation $\pi \in \P_{m \times n}$ we have $ p_\nu(Z) = p_\nu(\pi(Z)).$ Set $Y = \pi(X)$, where $X$ is a generic rank-$r$ matrix and $\pi \in \P_{m \times n}$. By the symmetry, $\pi_\nu(Y) = \pi_\nu(X)$ for every $\nu$. According to a well-known fact (e.g., see \cite{song2018permuted}), the solutions to the polynomial system of $mn$ equations $ p_\nu(Z) = p_\nu(Y), \, \nu \in [mn],$ are precisely all permutations of $X$. Let us add to this polynomial system the equations that force the solutions to have rank at most $r$. These ask that all $(r+1) \times (r+1)$ determinants of the solution are zero. Hence the new polynomial system is 
\begin{align}
& p_\nu(Z) = p_\nu(Y), \, \nu \in [mn], \label{eq:p(Z)=p(Y)}  \\
& \det(Z_{\I,\J}) = 0, \label{eq:det1}\\
 & \, \I \subset [m], \, \J \subset [n], \, \#\I = \# \J = r+1 \label{eq:det2},
\end{align} where $\#$ denotes cardinality of a set. The solutions of this system are all permutations of $X$ or rank at most $r$. But because $X$ is generic, these are precisely the permutations of $X$ described in Theorem \ref{thm:arbitrary-pi}. 

Let us interpret equations \eqref{eq:p(Z)=p(Y)} - \eqref{eq:det2} geometrically. First of all, the determinantal equations \eqref{eq:det1}-\eqref{eq:det2} define, inside the vector space $\Re^{m \times n}$ of all $m \times n$ real matrices, the space $\M_{r, m \times n}$ of matrices of rank at most $r$. On the other hand, for each $\nu$, the equation $p_\nu(Z) = p_\nu(Y)$ defines a hypersurface of $\Re^{m \times n}$. Thus the equations \eqref{eq:p(Z)=p(Y)} - \eqref{eq:det2} represent the intersection of $\M_{r, m \times n}$ with $mn$ hypersurfaces. Now, $\M_{r, m \times n}$ is an algebraic variety of dimension $r(m+n-r)$. Intuitively, the dimension of a variety is the least number of hypersurfaces that we need to intersect the variety with, in order to get a finite set of points. This suggests that we may be able to recover $X$ from $Y$ (in the sense of Theorem \ref{thm:arbitrary-pi}) from a polynomial system of the order of $\dim \M_{r, m \times n} = r(m+n-r)$ equations (we are not counting the determinantal constraints \eqref{eq:det1}-\eqref{eq:det2}), which {\textemdash depending on the rank $r$\textemdash} can be significantly fewer than the $mn$ equations \eqref{eq:p(Z)=p(Y)}. Our second result says that this is the case. 

\begin{thm} \label{thm:system}
Let $X$ be a generic matrix of rank $r$. The only rank-$r$ solutions to the polynomial system of equations defined by $r(m+n-r)+1$ generic linear combinations of all equations in \eqref{eq:p(Z)=p(Y)} are the permutations of $X$ described in Theorem \ref{thm:arbitrary-pi}. 
\end{thm}

\noindent In Theorem \ref{thm:system}, by a linear combination of all the equations in \eqref{eq:p(Z)=p(Y)} we mean an equation of the form $ \sum_{\nu \in [mn]} c_\nu p_\nu(Z) = 0.$ By a generic linear combination, we mean a linear combination with random coefficients $c_\nu$; this will be made precise inside the proof of the theorem. 

Finally, note that in the discussion preceding Theorem \ref{thm:system}, we did not count the determinantal equations \eqref{eq:det1}-\eqref{eq:det2}. This is because we can dispense with them by using the standard parametrization of $\M_{r, m \times n}$, in which $X$ is a product of an $m \times r$ matrix with an $r \times n$ matrix. 

\subsection{Ladder Matrices}

\subsubsection{Ladders}
We begin with the definition of the combinatorial structure called \emph{ladder}, following \cite{gorla2007mixed} and \cite{conca1995ladder}.

\begin{dfn}[ladder] \label{dfn:ladder}
A ladder is a subset $\Lc$ of $[m] \times [n]$ such that whenever $(i,j), (k,l) \in \Lc$, with $i \le k$ and $j \ge l$, then $(k,j), \, (i, l) \in \Lc$ (see Fig. \ref{fig:ladder-dfn}). For simplicity we will require that $(1,1), \, (m,n) \in \Lc$.
\end{dfn}

\begin{figure}[h]
\centering
\includegraphics[width=0.4\textwidth]{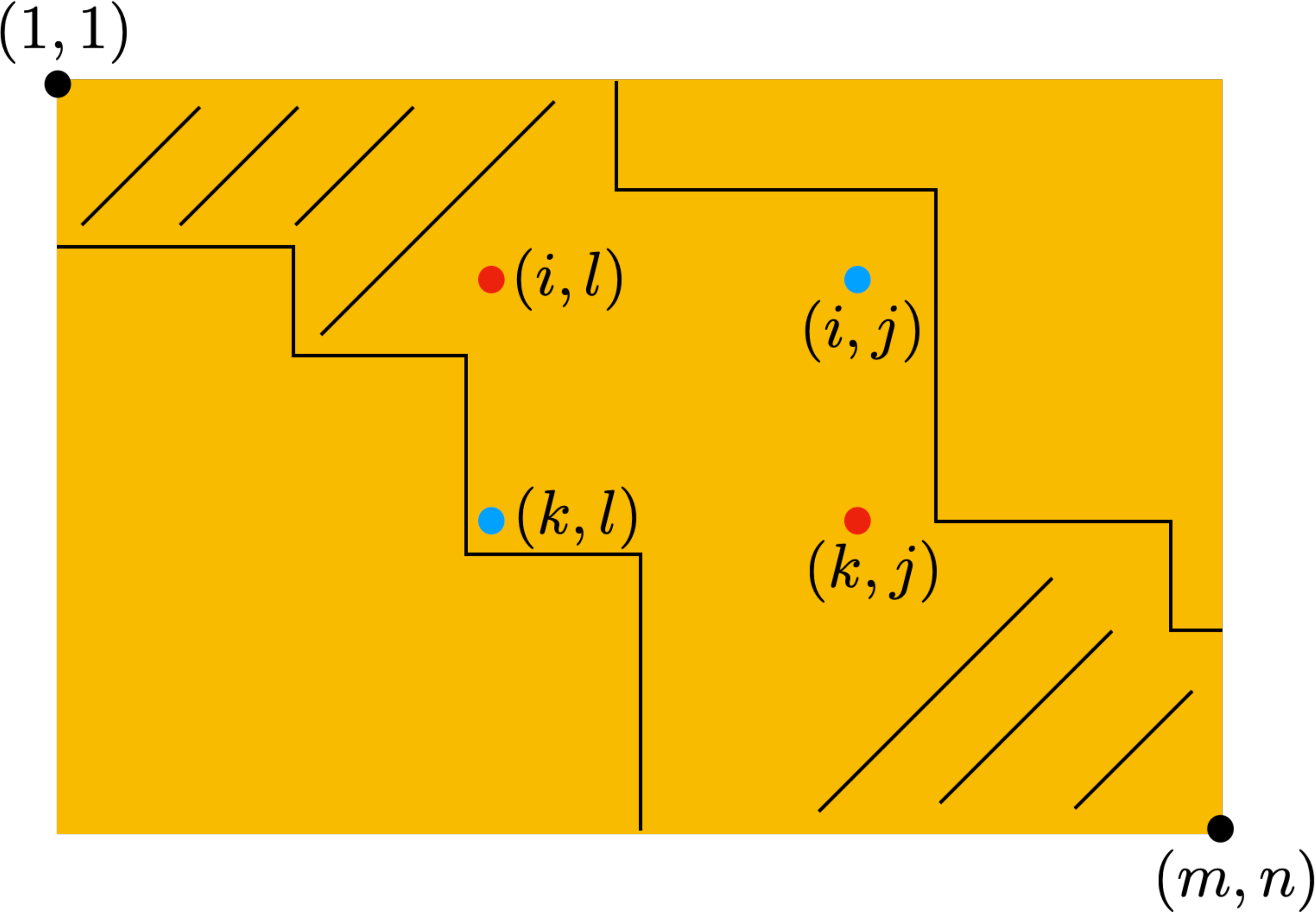}
\caption{The blue dots are inside the ladder and thus by definition they force the red dots to be inside the ladder as well (Definition \ref{dfn:ladder}).} \label{fig:ladder-dfn}	
\end{figure}

The following is a finer combinatorial description of a ladder, which is convenient for analyzing properties of the ladder and its associated objects.

\begin{prp} \label{prp:corners}
Let $\Lc \subset [m] \times [n]$ be a ladder. Then there exist positive integers $1 = a_1 < a_2 < \cdots < a_h \le m$,  $1 \le b_1< \cdots < b_h =n$,  $1 \le c_1 \le \cdots < c_k = m$, $1 = d_1 < d_2 < \cdots < d_k \le n$, such that $\Lc = \{(i,j): \, a_u \le i \le m, \, 1 \le j \le b_u, 1 \le i \le c_l, \, d_l \le j \le n$, for some $u \in [h], \, l \in [k]\}$. 
\end{prp}
 

\begin{dfn}[lower and upper outside corners] \label{dfn:ladder-corners}
In the notation of Proposition \ref{prp:corners}, the $(a_u,b_u)$'s are the upper outside corners, while the $(c_l,d_l)$'s are the lower outside corners of the ladder (see Fig. \ref{fig:ladder-corners}). 
\end{dfn}

\begin{figure}[h]
\centering
\includegraphics[width=0.4\textwidth]{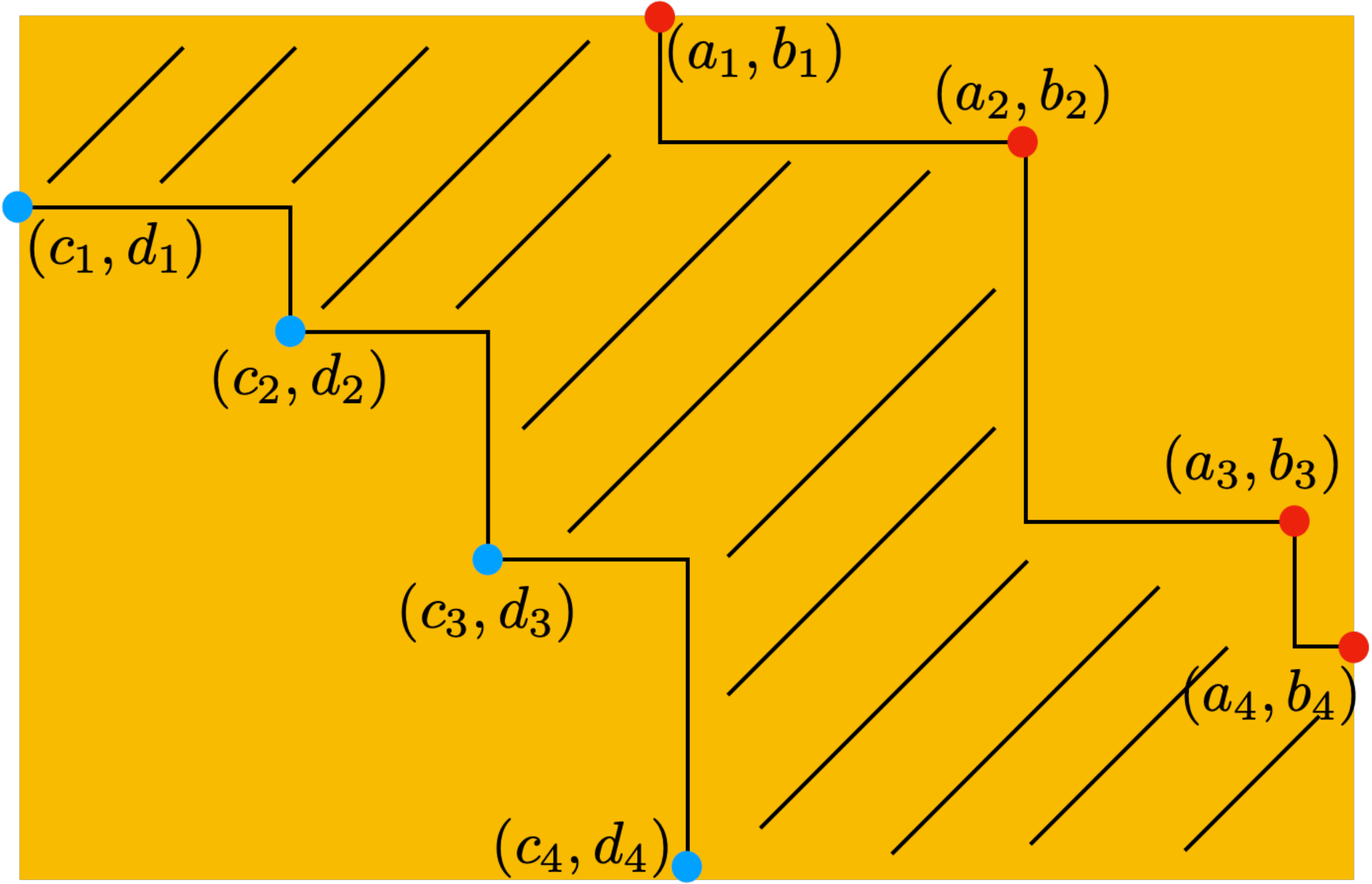}
\caption{The blue dots are the lower outside corners of the ladder; the red dots are the upper outside corners. The number of lower and upper outside corners need not be the same. (Definition \ref{dfn:ladder-corners}).} \label{fig:ladder-corners}	
\end{figure}

\begin{dfn}[shape of ladder]
Let $\Lc \subset [m] \times [n]$ be a ladder. The shape of $\Lc$ is an $m \times n$ matrix $S_{\Lc}=(s_{ij})$ of zeros and ones such that $s_{ij} = 1$ if and only if $(i,j) \in \Lc$. 
\end{dfn}

\subsubsection{Ladder Matrices and Ladder-Rank}
We introduce the notions of ladder matrix and ladder-rank.

\begin{dfn}[ladder matrix]
A ladder matrix is a matrix $X \in \Re^{m \times n}$ together with a ladder $\Lc \subset [m] \times [n]$, such that $X_{ij} = 0$ whenever $(i,j) \not\in \Lc$. 
\end{dfn}

\begin{dfn}[support in the ladder]
Let $X \in \Re^{m \times n}$ be a ladder matrix with associated ladder $\Lc$. We say that a submatrix $X_{\I,\J}$ of $X$, indexed by rows $\I \subset [m]$ and columns $\J \subset [n]$,  is supported in the ladder $\Lc$, if the rectangle $\I \times \J$ is contained in the ladder, i.e. $\I \times \J \subset \Lc$. 
\end{dfn}

\begin{dfn}[ladder-rank] \label{dfn:ladder-rank}
Let $X \in \Re^{m \times n}$ be a ladder matrix with associated ladder $\Lc$. We say that $X$ has ladder-rank $r$, denoted by $\rank_{\Lc}(X)$, if $r$ is the largest positive integer such that there exists an invertible $r \times r$ submatrix of $X$ with support in $\Lc$. Equivalently, $ {\rank_{\Lc}}(X) = \max \{ \rank (X_{\I,\J}): \, \I \subset [m], \, \J \subset [n], \, \I \times \J \subset \Lc \}$.
\end{dfn}

\noindent The ladder-rank of a ladder matrix $X$ bounds from below the ordinary rank of $X$, and coincides with it, whenever the ladder $\Lc$ is trivial, that is $\Lc = [m] \times [n]$. Intuitively, the ladder-rank aims at capturing the linear dependence patterns of the matrix that occur within the region of the ladder. Evidently, the ladder-rank is constrained by the shape of the ladder $\Lc$. In particular, it is bounded from above by the largest size of a square submatrix with support in the ladder. We denote this integer by $r(\Lc)$. We say that $X$ has bounded ladder-rank if $\rank_{\Lc}(X) < r(\Lc)$. 

\begin{remark}
In Definition \ref{dfn:ladder-rank} we gave a global definition of the ladder-rank. Instead, we could have considered it as a function of the region inside the ladder, for instance involving determinants of different sizes, as permitted by the shape of the ladder. This would allow the ladder matrix to have different ladder-rank in different regions. On the level of ideals in commutative algebra this has been studied in \cite{gonciulea2000mixed, gorla2007mixed}. From a data science point of view, it would result in a more subtle notion, able to capture finer variations in the data, and it is only for simplicity that we do not further develop it here. 
\end{remark}

\subsubsection{Ladder-Rank Preserving Permutations} 

Let $X \in \Re^{m \times n}$ be a ladder matrix with associated ladder $\Lc \subset [m] \times [n]$ and ladder-rank $r$. So there is an $r \times r$ non-zero determinant supported in the ladder, while all larger determinants, if any, are zero. 

We consider permutations that move only the entries inside the ladder:

\begin{dfn}[shape-preserving permutations]
We denote by $\P_{\Lc}$ the set of permutations $\pi \in \P_{m \times n}$ such that $\pi(S_{\Lc}) = S_{\Lc}$, where $S_{\Lc}$ is the shape of $\Lc$.
\end{dfn}

There are two obvious families of permutations that preserve the ladder-rank. For the first one, let $\Lc_{<r}$ be the subset of $\Lc$ consisting of all $(i,j) \in \Lc$, such that $(i,j)$ is not contained in any $r \times r$ square contained in $\Lc$. It is clear that every permutation $\pi \in \P_{\Lc}$ that moves only elements of $\Lc_{<r}$ is ladder-rank preserving, because it does not affect any $r \times r$ determinant. Note that when the ladder structure is trivial, i.e. $\Lc = [m] \times [n]$, then $\Lc_{<r}$ is the empty set. 

To describe the second family, we introduce ladder-rows and ladder-columns, which are horizontal and vertical slices of $X$ along the ladder:

\begin{dfn}[ladder-row, ladder-column]
For any $i \in [m]$ let $j_1$ and $j_s$ be minimal and maximal, respectively, such that $(i,j_1)$ and $(i,j_s)$ are in $\Lc$. We define the $i$th ladder-row of $X$ to be the ordered set $(x_{ij_1},\dots,x_{ij_s})$ of entries of $X$. The $j$th ladder-column of $X$ is defined similarly. 
\end{dfn}

\noindent The second family of obvious ladder-rank preserving permutations consists of permutations that only permute ladder-rows and ladder-columns. This follows from Definition \ref{dfn:ladder-rank} and by noting that the transposition of any two ladder-rows does not change the rank of the largest submatrix with support in the ladder that contains them (similarly for ladder-columns). 

\subsubsection{Recovery of ladder matrices}

Let $\Lc \subset [m] \times [n]$ be a ladder. Inside the vector space $\Re^{m \times n}$ of all $m \times n$ real matrices, it induces the linear subspace $\Re^{m \times n}_{\Lc}$ of ladder matrices $X$ with associated ladder $\Lc$. Another way to think of $\Re^{m \times n}_{\Lc}$ is as the set of all sparse matrices with support on $\Lc$. As a vector space, $\Re^{m \times n}_{\Lc}$ is isomorphic to $\Re^{\# \Lc}$. For any positive integer $r$, the set $\M_{r,\Lc}$ of all ladder matrices $X$ with $\rank_{\Lc}(X) \le r$ is an algebraic variety of $\Re^{m \times n}_{\Lc}$, called a \emph{ladder determinantal variety} \cite{narasimhan1986irreducibility,conca1995ladder}. Indeed, it is defined by requiring that $\det(X_{\I,\J}) = 0$ for any square $\I \times \J \subset \Lc$ of size $(r+1) \times (r+1)$. Note how the ladder determinantal variety $\M_{r,\Lc}$ reduces to the well-known  \emph{determinantal variety} $\M_{r,m \times n}$ of matrices of bounded rank, when the ladder is trivial. 

A classical result of invariant theory and commutative algebra \cite{bruns2006determinantal} asserts that the variety $\M_{r,m \times n}$ is \emph{irreducible}, that is $\M_{r,m \times n}$ is not the union of two or more proper subvarieties of it. This is the crucial property that allows us to talk about \emph{a generic matrix of rank $r$}. Similarly, it was proved in \cite{narasimhan1986irreducibility} that the ladder determinantal variety $\M_{r,\Lc}$ is irreducible. Hence, we can talk about \emph{a generic ladder matrix of ladder-rank $r$}. Our next result, which generalizes Theorem \ref{thm:arbitrary-pi}, says that for a generic ladder matrix of ladder-rank $r$, the only ladder-rank preserving permutations are the expected ones:

\begin{thm} \label{thm:pi-ladder}
Let $X$ be a generic ladder-matrix of ladder-rank $r$ with associated ladder $\Lc \subset [m] \times [n]$. Let $\pi \in \P_{\Lc}$ be a permutation. Then $\rank_{\Lc} (\pi(X)) > \rank_{\Lc} (X)$, unless 1) $\pi$ permutes ladder-rows and ladder-columns of $X$ or 2) $\pi$ permutes entries of $X$ that are not contained in any $r \times r$ submatrix with support in $\Lc$ or 3) $\pi$ is such that $\pi(X) = X^\top$ in case $m=n$ and the shape of $\Lc$ is invariant under matrix transposition. 
\end{thm} 

\noindent Note that when the ladder is trivial, i.e. $\Lc = [m] \times [n]$, Theorem \ref{thm:pi-ladder} reduces to Theorem \ref{thm:arbitrary-pi}. 

Our final result is a generalization of Theorem \ref{thm:system}. This time we have to consider symmetric polynomials that include only the variables inside the ladder: 
\begin{align*}
p_{\Lc,\nu}(Z) &= \sum_{(i,j) \in \Lc} z_{ij}^\nu
\end{align*}

\noindent Now, $Y = \pi(X)$ for some $\pi \in \P_\Lc$ and $\rank_\Lc(Y) \le r$, if and only if $Y$ is a root of the polynomial system of equations 
\begin{align}
&\hat{p}_{\Lc,\nu}(Z) := p_{\Lc,\nu}(Z)-p_{\Lc,\nu}(X)=0, \label{eq:ladder-system-perm-begin}\\
& \forall \nu \in [\#\Lc]; \label{eq:ladder-system-perm-end}\\
&\det(Z_{\I,\J})=0, \label{eq:ladder-system-rank-begin}\\
& \forall \, \I \times \J \subset \Lc, \, \#\I=\#\J=r+1. \label{eq:ladder-system-rank-end}
\end{align} The same principles that led us to Theorem \ref{thm:system} can be applied here: instead of considering all $\#\Lc$ equations $\hat{p}_{\Lc,\nu}(Z)=0$, it is enough to take $\dim(\M_{r,\Lc})+1$ generic linear combinations of them. The dimension of ladder determinantal varieties has been computed in \cite{herzog1992grobner} and \cite{gorla2007mixed} and it can be expressed in more than one ways. Here we quote a beautiful theorem in \cite{gorla2007mixed}: 

\begin{prp}(Theorem 1.15 in \cite{gorla2007mixed}) \label{prp:Gorla}
Consider the ladder $\Lc' \subset \Lc$ that has the same upper outside corners as $\Lc$ and lower outside corners $(c_l',d_l') = (c_l-r,d_l+r)$ for $l \in [k]$ (see Definition \ref{dfn:ladder-corners}). Then $$\dim(\M_{r,\Lc}) = \#\Lc  - \#\Lc'.$$ 
\end{prp}

\begin{remark} When the ladder is trivial ($\Lc = [m] \times [n]$), it has only one lower and one upper outside corner, $(m,1)$ and $(1,n)$, respectively. In that case, the ladder $\Lc'$ of Proposition \ref{prp:Gorla} has lower outside corner $(m-r,r+1)$ and upper outside corner $(1,n)$. This is the top right block of $[m] \times [n]$ indexed by the first $(m-r)$ rows and the last $(n-r)$ columns. Thus the dimension of this ladder determinantal variety is $mn-(m-r)(n-r) = r(m+n-r)$. Indeed, this is precisely the dimension of $\M_{r,m \times n}$. 
\end{remark}

\begin{thm} \label{thm:ladder-system}
Let $X$ be a generic ladder matrix of ladder-rank $r$ with associated ladder $\Lc$. Then the solutions to the polynomial system made up of equations \eqref{eq:ladder-system-rank-begin}-\eqref{eq:ladder-system-rank-end}, together with $\dim(\M_{r,\Lc})+1$ linear combinations of equations \eqref{eq:ladder-system-perm-begin}-\eqref{eq:ladder-system-perm-end} with coefficients chosen generically in $\Re$, are the ladder-rank preserving permutations of $X$ given in Theorem \ref{thm:pi-ladder}. 
\end{thm}

\section{Proofs} \label{section:Proofs}

The proofs of Theorems \ref{thm:arbitrary-pi} and \ref{thm:pi-ladder} rely on results in the theory of Gr\"obner bases for determinantal \cite{sturmfels1990grobner} and ladder determinantal ideals \cite{narasimhan1986irreducibility}. The proofs of Theorems \ref{thm:system} and \ref{thm:ladder-system} involve a continuation of the dimension theory exposed in \cite{tsakiris2020algebraic}; they also require going deeper into commutative ring theory (\cite{matsumura1989commutative}). It is the aim of this article to make these arguments accessible to the audience not familiar with commutative algebra and algebraic geometry, without sacrificing any rigor. Towards that end, Appendices \ref{appendix:Ring-Theory} and \ref{appendix:GB} serve as a concise yet self-contained preparation for the reader. Other necessary background beyond \cite{tsakiris2020algebraic} is given along the way. 

Finally, we mention that Theorems \ref{thm:arbitrary-pi} and \ref{thm:system} are special cases of Theorems \ref{thm:pi-ladder} and \ref{thm:ladder-system}. While we could have only proved these latter two theorems, for pedagogical but also motivational reasons we do in fact the opposite: we prove Theorems \ref{thm:arbitrary-pi} and \ref{thm:system} which are concerned with the important case of matrices of bounded rank, and only give the necessary ingredients for the proofs of Theorems \ref{thm:pi-ladder} and \ref{thm:ladder-system} that follow very similar arguments. 

\subsection{Proof of Theorem \ref{thm:arbitrary-pi}}

Let $Z=(z_{ij})$ be an $m \times n$ matrix of indeterminates and $\mathscr{R} = \Re[Z]$ the set of all polynomials in variables $z_{ij}$ with real coefficients. We consider the following monomial order (Definitions \ref{dfn:monomial-order} and \ref{dfn:lexicographic-order}):

\begin{dfn}[diagonal order] \label{dfn:diagonal-order}
We define $\succ$ to be  the lexicographic order on $\mathscr{R}$ corresponding to $z_{11} \succ z_{12} \succ \cdots z_{1n} \succ z_{21} \succ \cdots \succ z_{2n} \succ \cdots \succ z_{mn}$. 
\end{dfn}

\noindent Such an order is known as a \emph{diagonal} order \cite{bruns2003grobner}, because the initial monomial (Definition \ref{dfn:initial-monomial}) of any determinant of $Z$ is the product of the variables along the main diagonal of the determinant. Precisely, let $\I=\{i_1< i_2 <\cdots<i_k\} \subset [m]$ and $\J=\{j_1<j_2<\cdots<j_k\} \subset [n]$ be subsets of row and column indices of size $k$. Then $$\operatorname{in}_{\succ} (\det(Z_{\I,\J})) = z_{i_1 j_1} z_{i_2 j_2} \cdots z_{i_k j_k}.$$ 
Now let us consider the ideal (Definition \ref{dfn:ideal-generated-by}) $\mathscr{I}$ of $\mathscr{R}$ generated by all $(r+1) \times (r+1)$ determinants of $Z$, that is $ \Is = (\det(Z_{\I,\J}): \, \I \subset [m], \, \J \subset [n], \, \#\I = \# \J = r+1)$. Recalling Definition \ref{dfn:GB}, the following is an important result in the field of commutative algebra, proved independently by different authors and different techniques (see Remark 5.4 in \cite{bruns2003grobner}): 

\begin{lem}\label{lem:GB-det}(\cite{narasimhan1986irreducibility,sturmfels1990grobner,caniglia1990ideals,ma1994minors,sturmfels2006combinatorial}) The set of all $(r+1) \times (r+1)$ determinants of $Z$ is a Gr\"obner basis under $\succ$ for the ideal $\Is$ that they generate. 
\end{lem}

Let $\pi \in \P_{m \times n}$ be a permutation that does not solely permute rows and columns of $Z$, and $\pi(Z) \neq Z^\top$ in case $m=n$. 

\begin{lem} \label{lem:Is}
There is an $(r+1)\times(r+1)$ determinant of $\pi(Z)$ which does not lie in $\Is$.
\end{lem}
\begin{proof}
On the contrary, suppose that all $(r+1) \times (r+1)$ determinants of $\pi(Z)$ lie in $\Is$. Thus for any $\I=\{i_1< i_2 <\cdots<i_k\} \subset [m]$ and $\J=\{j_1<j_2<\cdots<j_k\} \subset [n]$, we have that $\det(\pi(Z)_{\I,\J}) \in \Is$. So the initial monomial of $\det(\pi(Z)_{\I,\J})$ lies inside the initial ideal (Definition \ref{dfn:initial-ideal}) of $\Is$, that is $$\operatorname{in}_{\succ}(\det(\pi(Z)_{\I,\J})) \in \operatorname{in}_{\succ}(\Is).$$ 
By Lemma \ref{lem:GB-det} and Definition \ref{dfn:GB}, the initial ideal of $\Is$ is generated by the initial monomials of all $(r+1) \times (r+1)$ determinants of $Z$. By Lemma \ref{lem:monomial-ideal}, $\operatorname{in}_{\succ}(\det(\pi(Z)_{\I,\J}))$ must be divisible by some $\operatorname{in}_{\succ}(\det(Z_{\I',\J'}))$ for $\I'$ and $\J'$ subsets of row and column indices, respectively, of size $r+1$. But both $\operatorname{in}_{\succ}(\det(\pi(Z)_{\I,\J}))$ and $\operatorname{in}_{\succ}(\det(Z_{\I',\J'}))$ are monomials of degree $r+1$. Hence 
\begin{align}
 \operatorname{in}_{\succ}(\det(\pi(Z)_{\I,\J})) = \operatorname{in}_{\succ}(\det(Z_{\I',\J'})). \label{eq-lem:in-det}
 \end{align}

By the assumption on $\pi$, and composing it if necessary with a suitable permutation that only permutes rows and columns, we may further assume $$\pi(Z)_{1,1} = z_{11}, \, \, \, \pi(Z)_{2,2} = z_{12}.$$ Take $\I=\J=[r+1]$. Since $z_{11}$ and $z_{12}$ are the two largest variables under the lexicographic order $\succ$, the initial monomial of $\det(\pi(Z)_{\I,\J})$ must be divisible by $z_{11} z_{12}$. However, for no determinant of $Z$ does $z_{11} z_{12}$ divide any monomial in its support (Definition \ref{dfn:support}). Hence there are no $\I', \J'$ that satisfy $\eqref{eq-lem:in-det}$.
\end{proof}

Now, $\pi^{-1}$ is an \emph{automorphism} of $\Re^{m \times n}$ \textemdash{it simply permutes the coordinates}. Thus the image $\pi^{-1}(\M_{r, m \times n})$ of the variety $\M_{r, m \times n}$ under $\pi^{-1}$ is another variety (in fact, it is \emph{isomorphic} to $\M_{r, m \times n}$). 

\begin{lem}\label{lem:equations-pi-inverse-M}
The polynomials that define the algebraic variety $\pi^{-1}(\M_{r, m \times n})$ are the $(r+1) \times (r+1)$ determinants of $\pi(Z)$.
\end{lem}
\begin{proof}
 A matrix $Y$ lies in $\pi^{-1}(\M_{r, m \times n})$ if and only if $Y = \pi^{-1}(X)$ for some $X \in \M_{r, m \times n}$. Thus $\pi(Y)$ has rank at most $r$. Hence all of its $(r+1) \times (r+1)$ determinants are zero. 
 \end{proof}

The next fact that we need is another important result of invariant theory and commutative algebra:

\begin{lem}(e.g., \cite{bruns2006determinantal}) \label{lem:Is-kernel}
Let $p \in \Rs$ such that $p(X) = 0$ for every $X \in \M_{r,m \times n}$. Then $p \in \Is$. 
\end{lem}

\noindent Lemmas \ref{lem:Is}-\ref{lem:Is-kernel} imply that the intersection of the two varieties $\M_{r, m \times n} $ and $\pi^{-1}(\M_{r, m \times n})$ is a properly contained subvariety of $\M_{r, m \times n} $:
\begin{lem}\label{lem:proper-containment}
$ \big[\M_{r, m \times n} \cap \pi^{-1}(\M_{r, m \times n})\big] \subsetneq \M_{r, m \times n}.$
\end{lem}
\begin{proof}
 By Lemma \ref{lem:Is} there is an $(r+1) \times (r+1)$ determinant $p$ of $\pi(Z)$ which does not lie in $\Is$. Hence by Lemma \ref{lem:Is-kernel}, there must exist some $X \in \M_{r,m \times n}$ such that $p(X) \neq 0$. But then $X$ is not in $\pi^{-1}(\M_{r, m \times n})$ (Lemma \ref{lem:equations-pi-inverse-M}). 
 \end{proof}
 
The set $\mathcal{Y}_\pi = \M_{r, m \times n} \cap \pi^{-1}(\M_{r, m \times n})$ is an algebraic variety of $\Re^{m \times n}$; it is defined by the union of the equations that define $\M_{r, m \times n}$ and $\pi^{-1}(\M_{r, m \times n})$. By Lemma \ref{lem:proper-containment}, it is properly contained in $\M_{r, m \times n}$. It can be thought of as the set of pathological matrices of rank at most $r$, which, when permuted by $\pi^{-1}$, their rank still does not exceed $r$. Let $\mathcal{U}_\pi$ be the complement of $\Ys_\pi$ in $\M_{r, m \times n}$. Since the closed sets of the Zariski topology are the algebraic varieties, $\mathcal{U}_\pi$ is a non-empty open set. Now let $\mathcal{U} = \cap_{\pi} \mathcal{U}_\pi,$ where the intersection is taken over all permutations that do not only permute rows and columns. Since the intersection of finitely many non-empty open sets of an irreducible variety is again open and non-empty (see the discussion in \S 3.1 in \cite{tsakiris2020low}), so will be $\mathcal{U}$. Moreover, by its construction, we have that $\rank(\pi(X)) > r$ for every $X \in \mathcal{U}$ and every $\pi$ that does not only permute rows and columns. 

\subsection{Proof of Theorem \ref{thm:system}}

We have seen that the algebraic variety $\M_{r, m \times n}$ is defined by all $(r+1) \times (r+1)$ determinants of $Z$. As in the proof of Theorem \ref{thm:arbitrary-pi}, we call $\mathscr{I}$ the ideal 
that these determinants generate. By definition, $p \in \mathscr{I}$ if and only if $p$ is a linear combination of the determinants, with the coefficients being other polynomials. We see that $X \in \M_{r, m \times n}$ if and only if $p(X)=0$ for every $p \in \Is$. Thus we will think of $\M_{r, m \times n}$ as being defined by the ideal $\Is$, instead of just the determinants. We indicate this by writing $\M_{r, m \times n}=\V(\Is)$. 

The ideal $\Is$ has been well studied. Recalling Definitions \ref{dfn:height}, \ref{dfn:ring-dimension}, \ref{dfn:quotient-ring} and \ref{dfn:Cohen-Macaulay ring}, we have: 

\begin{lem}(\cite{mount1967remark,eagon1969examples}) \label{lem:I-height}
$\height(\Is) = (m-r)(n-r).$
\end{lem}

\begin{lem}(\cite{northcott1963some,mount1967remark,hochster1971cohen}) \label{lem:I-prime}
The ideal $\Is$ is prime. 
\end{lem}

\begin{lem}(\cite{hochster1971cohen}) \label{lem:I-CM}
The quotient ring $\bar{\Rs} = \Rs / \Is$ is Cohen-Macaulay. 
\end{lem}

\noindent We fix $X \in \M_{r, m \times n}$ and consider the polynomials $$\hat{p}_\nu(Z) = p_\nu(Z)-p_\nu(X), \, \, \, \nu \in [mn],$$ inside the ring $\Rs=\Re[Z]$. 

\begin{lem} \label{lem:hat-p-regular}
$\hat{p}_1,\dots,\hat{p}_{mn}$ are a $\Rs$-regular sequence.
\end{lem}
\begin{proof}
The polynomials $p_1,\dots,p_{mn}$ are a $\Rs$-regular sequence (\cite{conca2009regular}). 
Consider the weight-order on $\Rs$ induced by the vector $w=(1,\dots,1)$; see Appendix E in 
\cite{tsakiris2020algebraic}. Under that order, the initial form $\In_w(\hat{p}_\nu)$ of $\hat{p}_\nu$ is $p_\nu$. Now the statement follows from a general fact (Proposition 1.2.12 in \cite{Aldo-book}), according to which, if the sequence of initial forms $\In(g_1),\dots,\In(g_k)$ of a sequence of polynomials $g_1,\dots,g_k$ is a regular sequence, then $g_1,\dots,g_k$ is also a regular sequence (see additionally Proposition E2 in \cite{tsakiris2020algebraic}).  
\end{proof}

\noindent Let $\Ls$ be the ideal that the $\hat{p}_\nu$'s generate. 

\begin{lem} \label{lem:dimR/J}
$\dim (\Rs / \Ls)=0.$
\end{lem}
\begin{proof}
Since the zero ideal $(0)$ of $\Rs$ is a prime ideal, by Proposition \ref{prp:dimension-formula} we have $\height(\Ls) \le \dim (\Rs)$. The dimension of a polynomial ring is equal to the number of its variables. Thus $\dim (\Rs) = mn$ and so $\height(\Ls) \le mn$. By Lemma \ref{lem:hat-p-regular}, $\Ls$ contains a regular sequence of length $mn$; thus $\grade(\Ls) \ge mn$. Then Proposition \ref{prp:grade-height} gives $mn \le \grade(\Ls) \le \height(\Ls) \le mn$. Hence $\height(\Ls) = mn$. Now Proposition \ref{prp:dimension-formula} gives $\dim(\Rs/\Ls) = \dim (\Rs) - \height(\Ls) = mn-mn=0$.
\end{proof}

\noindent From Definition \ref{dfn:dimension-variety} and Lemma \ref{lem:dimR/J} the variety $\V(\Ls)$ defined by $\Ls$ has dimension zero, i.e. $\V(\Ls)$ is a finite set of points. The following is a consequence of a well-known general fact (e.g., see \cite{song2018permuted}):

\begin{lem}\label{lem:V(L)}
The variety $\V(\Ls)$ consists of the $(mn)!$ matrices obtained by permuting the entries of $X$.
\end{lem}

 \noindent It follows that the intersection $\V(\Is) \cap \V(\Ls)$ of the two varieties $\V(\Is)$ and $\V(\Ls)$ consists of all permutations of $X$ of rank at most $r$. The ideal defining this intersection, as a subvariety of $\Re^{m \times n}$, is the sum (Definition \ref{dfn:sum-of-ideals}) $\Is + \Ls$ of the two ideals, that is, $\V(\Is) \cap \V(\Ls) = \V(\Is + \Ls)$. Equivalently, by Proposition \ref{prp:second-isomorphism-theorem}, the ideal $\bar{\Ls}$ of $\bar{\Rs}=\Rs/\Is$ defines $\V(\Is) \cap \V(\Ls)$ as a subvariety of $\V(\Is)$. 

\begin{lem} \label{lem:grade-barJ}
$\grade(\bar{\Ls}) = \dim (\Rs/\Is)$. 
\end{lem}
\begin{proof}
By the proof of Lemma \ref{lem:dimR/J} the ideal $\Ls$ has maximal height equal to $mn$. A fortiori $\height(\Ls+\Is)=mn$. Since the zero ideal $(0)$ of $\Rs$ is a prime ideal, Proposition \ref{prp:dimension-formula} gives $\dim (\Rs / \Is+\Ls)=0$. Set $\bar{\Rs} = \Rs/ \Is$ and set $\bar{\Ls}$ the ideal of $\bar{\Rs}$ induced by $\Ls$. By Proposition \ref{prp:second-isomorphism-theorem} we have $\dim (\bar{\Rs} / \bar{\Ls}) = 0$. Since $\Is$ is a prime ideal (Lemma \ref{lem:I-prime}), once again Proposition \ref{prp:dimension-formula}  gives $\height(\bar{\Ls})=\dim(\bar{\Rs})$. Now, $\bar{\Rs}$ is a Cohen-Macaulay ring (Lemma \ref{lem:I-CM}), so Definition \ref{dfn:Cohen-Macaulay ring} gives $\grade(\bar{\Ls}) = \dim (\bar{\Rs})$. 
\end{proof}

With this preparation, the proof will partly follow from an inductive application of a general fact:

\begin{lem} \label{lem:Conca}
Let $\Ts$ be a Noetherian ring that contains an infinite field $\mathbb{K}$. Let $\As=(t_1,\dots,t_n)$ be an ideal of $\Ts$ with $\grade(\As)>0$. Then a linear combination $t=c_1 t_1+\cdots+c_n t_n$ with the $c_i$'s chosen generically in $\mathbb{K}$ is $\Ts$-regular.
\end{lem}
\begin{proof}
On the contrary, suppose that $t=c_1 t_1+\cdots+c_n t_n$ is a zero-divisor of $\Ts$ for every choice in $\mathbb{K}$ of the $c_i$'s. In particular, for every non-zero $a \in \mathbb{K}$ the element $t_a = t_1 + a t_2 + \cdots + a^{n-1} t_n$ is a zero-divisor. Hence for every $0 \neq a \in \mathbb{K}$ the element $t_a$ lies in the union of the associated primes of $\Ts$ (Proposition \ref{prp:associated-primes}). Since there are finitely many associated primes (Proposition \ref{prp:associated-primes}), and since $\mathbb{K}$ is infinite, there is an associated prime $\Ps \in \Ass(\Ts)$, which contains infinitely many $t_a$'s. So let $a_1,\dots,a_n$ be distinct elements of $\mathbb{K}$ and $t_{a_1},\dots,t_{a_n} \in \Ps$.  We can write each $t_{a_j}$ as the inner product of the row vector $[t_1 \, \, t_2 \, \, \cdots \, \, t_n]$ with the column vector $\gamma_j=[1 \,  \, a_j \,  \, a_j^2 \, \, \cdots \, \, a_j^{n-1}]^\top$. Stacking each $\gamma_j$ in the $j$th column of a matrix $A$, we have the vector equality $[t_{a_1} \, \cdots \, t_{a_n}] = [t_1 \, \cdots \, t_n] A$ over $\Ts$. The matrix $A$ has size $n \times n$ and it is invertible, because it is the Vandermode matrix of $n$ distinct elements of $\mathbb{K}$. Thus the $t_i$'s can be expressed in terms of the $t_{a_j}$'s. Since the $t_{a_j}$'s lie in $\Ps$, so will the $t_i$'s. Hence $\As$ is contained in $\Ps$. But $\Ps$ is an associated prime, and as such, it consists of zero-divisors. Thus $\As$ also consists of associated primes. It follows that $\grade(\As)=0$; a contradiction. 

Let $\W$ be the vector space spanned by $t_1,\dots,t_n$ over $\mathbb{K}$. The above argument shows that $\mathcal{W}$ is not contained in any associated prime $\Ps$ of $\Ts$. Because $\Ps$ is an ideal and $\mathbb{K}$ is contained in $\Ts$, $\Ps$ is naturally a vector space over $\mathbb{K}$. Hence $\W \cap \Ps$ is a proper $\mathbb{K}$-subspace of $\W$. Moreover, since $\mathbb{K}$ is infinite, $\W$ is not the union $\mathscr{Y}= \cup_{\Ps \in \Ass(\Ts)}  \W \cap \Ps$ of its finitely many proper subspaces $\W \cap \Ps$. Since $\mathscr{Y}$ is an algebraic variety of $\W$ (Theorem 54 in \cite{tsakiris2017filtrated}), the complement $\mathscr{U}$ of $\mathscr{Y}$ in $\W$ is a non-empty open set. Since $\W$ is irreducible as an algebraic variety (Proposition 53 in \cite{tsakiris2017filtrated}), $\mathscr{U}$ is a dense open set of $\W$. By construction, any element of $\mathscr{U}$ is a linear combination of the $t_i$'s with coefficients in $\mathbb{K}$ and it is $\Ts$-regular.
\end{proof}

\noindent Set $\bar{\Rs}_0 = \bar{\Rs}=\Rs/\Is$ and $\bar{\Ls}_0=\bar{\Ls}$. Recalling the Definition \ref{dfn:equidimensional-ring} of an equidimensional ring, a consequence of Lemma \ref{lem:Conca} is:

\begin{lem} \label{lem:generic-combination-p_nu}
Let $h_1 =\sum_{\nu \in [mn]} c_\nu \hat{p}_\nu $ be a linear combination of the $\hat{p}_{\nu}$'s with the coefficients $c_i$ chosen generically in $\Re$. Set $\bar{\Rs}_1 = \bar{\Rs}_0 / (\bar{h}_1)$ and $\bar{\Ls}_1$ the ideal of $\bar{\Rs}_1$ induced by the ideal $\Ls$. Then $\bar{\Rs}_1$ is equidimensional with $\dim(\bar{\Rs}_1) = \dim(\bar{\Rs}_0)-1$. Moreover,  $\grade(\bar{\Ls}_1) = \grade(\bar{\Ls}_0)-1$.
\end{lem}
\begin{proof}
The grade of $\bar{\Ls}_0$ is equal to $\dim (\bar{\Rs}_0)$ (Lemma \ref{lem:grade-barJ}). Since $\height(\Is) = (m-r)(n-r)$ (Lemma \ref{lem:I-height}), and since $\dim (\Rs) = mn$, Proposition \ref{prp:dimension-formula} gives $\dim (\bar{\Rs}_0) = r(m+n-r)$. Hence $\grade(\bar{\Ls}_0) = r(m+n-r)>0$. 

As an ideal of $\bar{\Rs}_0$, $\bar{\Ls}_0$ is generated by the equivalence classes $\bar{\hat{p}}_\nu$ of the $\hat{p}_\nu$'s in $\bar{\Rs}_0$. Moreover, $\Re$ is contained in the ring $\bar{\Rs}_0$. Finally, $\bar{\Rs}_0$ is a Noetherian ring, because $\Rs$ is. Hence Lemma \ref{lem:Conca} applies: a linear combination $c_1 \bar{\hat{p}}_1 + \cdots + c_{mn} \bar{\hat{p}}_{mn}$ with coefficients $c_i$ chosen generically in $\Re$, will give an $\bar{\Rs}_0$-regular element. This is the equivalence class in $\bar{\Rs}_0$ of $h_1 = c_1 \hat{p}_1 + \cdots + c_{mn} \hat{p}_{mn} \in \Rs$. 

The fact that $\bar{\Rs}_1$ is equidimensional of dimension $\dim(\bar{\Rs}_0)-1$ will follow from Proposition \ref{prp:equidimensional-drop-1}, once we establish that $\Is + (h_1)$ is not the entire ring $\Rs$. So suppose $\Is+(h_1)=\Rs$. Then the variety $\V(\Is+(h_1))$ is empty, because $1 \in \Rs=\Is+(h_1)$. On the other hand, $\hat{p}_\nu(X) =0$ for every $\nu$ and hence also $h_1(X)=0$. Thus $X \in \V((h_1))$. Moreover, $X$ has rank at most $r$, so $X \in \V(\Is)$. Thus $X \in \V(\Is) \cap \V((h_1)) = \V(\Is+(h_1))$, contradicting the fact that $\Is+(h_1) = \Rs$. Finally, $\grade(\bar{\Ls}_1) = \grade(\bar{\Ls}_0)-1$ follows from Proposition \ref{prp:grade-regular}.
\end{proof}

For the same reasons as in the proof of Lemma \ref{lem:generic-combination-p_nu},  Lemma \ref{lem:Conca} can be applied inductively to the ring $\bar{\Rs}_1$ and its ideal $\bar{\Ls}_1$, yielding a sequence of rings $\bar{\Rs}_i$ and ideals $\bar{\Ls}_i$, which satisfy
\begin{align*}
 \bar{\Rs}_i &= \bar{\Rs}_{i-1} / (\bar{h}_i) \\
 \dim ({\Rs}_i) &= \dim ({\Rs}_{i-1})-1 \\
 \grade ({\Ls}_i) &= \grade ({\Ls}_{i-1})-1,
 \end{align*} with $h_i$ a generic linear combination of the $\hat{p}_\nu$'s with coefficients in $\Re$. With $d=r(m+n-r),$ the last member of the sequence is the ring $\bar{\Rs}_{d}$ and its ideal $\bar{\Ls}_{d}$, of zero dimension and zero grade, respectively. Our last touch in the proof is:
 
 \begin{lem}
 Suppose $X$ is a generic matrix of rank $r$. Let $h_{d+1}$ be a linear combination of the $\hat{p}_\nu$'s with coefficients chosen generically in $\Re$. Then the variety $\V(\Is + (h_1,\dots,h_{d+1}))$ consists of all matrices of the form $\Pi_1 X \Pi_2$, for all permutations $\Pi_1$ and $\Pi_2$, and it additionally contains $X^\top$ if $m =n$.
 \end{lem}
 \begin{proof}
 By Proposition \ref{prp:second-isomorphism-theorem}, we have $$\bar{\Rs}_d \cong  \Rs / \Is + (h_1,\dots,h_d).$$ \noindent Since $\dim(\bar{\Rs}_d)=0$, the variety $\V(\Is + (h_1,\dots,h_d))$ is zero-dimensional. As such, it consists of a finite set of points, which includes the set $\V(\Is+\Ls)$ of all permutations $\pi(X)$ of $X$ with $\rank(\pi(X)) \le r$. Let $\mathscr{Y}=\{Y_1,\dots,Y_\ell\}$ be the complement of $\V(\Is+\Ls)$ in $\V(\Is + (h_1,\dots,h_d))$. We will show that for a generic choice of the coefficients $c_\nu$, none of the $Y_i's$ is a root of the polynomial $h_{d+1} = \sum_{\nu \in [mL]} c_\nu \hat{p}_\nu$, by which we will be done. 
 
 Form an $\ell \times (mn)$ matrix $B=(b_{ij})$, with $b_{ij}=\hat{p}_j(Y_i)$. We have to show that for a generic vector $c \in \Re^{mn}$, none of the entries of the vector $Bc$ is equal to zero. Equivalently, we have to show that for a generic $c$, the point $b=Bc \in \Re^\ell$ does not lie in any of the $\ell$ coordinate hyperplanes of $\Re^{\ell}$; that is $b \notin \cup_{i \in [\ell]} \H_i$, with $\H_i$ the hyperplane of $\Re^\ell$ defined by requiring the $i$th coordinate to be zero. 
 
 Now, for each $i \in [\ell]$, $Y_i $ is a matrix which is not a permutation of $X$. Since a matrix $Y$ is a permutation of $X$ if and only if $\hat{p}_\nu(Y)=0$ for every $\nu \in [mn]$ (Lemma \ref{lem:V(L)}), for each $i$ there exists some $\nu_i \in [mn]$ such that $\hat{p}_{\nu_i}(Y_i)\neq 0$. Thus, none of the rows of $B$ is equal to zero. Hence, the column-space $\W$ of $B$ is not contained in any coordinate hyperplane $\H_i$. So $\W_i = \W \cap \H_i$ is a proper linear subspace of $\W$. As in the last paragraph of the proof of Lemma \ref{lem:Conca}, $\cup_{i \in [\ell]} \W_i$ is an algebraic variety properly contained in $\W$, and so its complement in $\W$ is a non-empty open (and thus dense) set $\mathscr{U}$ of $\W$. By construction, $\mathscr{U}$ has the property that for any $b \in \mathscr{U}$, none of the coordinates of $b$ is equal to zero. 
 
 Since $\mathscr{U}$ is non-empty and open, it is defined by the non-vanishing of certain non-zero polynomial equations in $\ell$ variables. That is, there is a set of polynomials $g_1,\dots,g_s$ in the polynomial ring $\Re[\underline{y}]=\Re[y_1,\dots,y_\ell]$, such that $b \in \mathscr{U}$ if and only if not all $g_1(b),\dots,g_s(b)$ are zero. Let $\Re[\underline{t}]=\Re[t_1,\dots,t_{mn}]$ be another polynomial ring. Viewing $\underline{t}$ as a column vector $[t_1,\dots,t_{mn}]^\top$, let us  replace the variables $\underline{y}$ in the $g_i$'s by $B \underline{t}$. This gives us a set of polynomials $g_1(B\underline{t}), \dots, g_s(B\underline{t})$ of $\Re[\underline{t}]$. These polynomials define an open set $\mathscr{U}'$ of $\Re^{mn}$, as the set of all $c \in \Re^{mn}$ for which not all $g_1(Bc), \dots, g_s(Bc)$ are zero. The open set $\mathscr{U}'$ is non-empty because $\mathscr{U}$ is non-empty: since $\mathscr{U}$ lies in the column-space of $B$, for every $b \in \mathscr{U}$ there is a $c \in \Re^{mn}$ with $b = Bc$, and so $c \in \mathscr{U}'$. It follows that $\mathscr{U}'$ is a dense open set of $\Re^{mn}$ with the property that for every $c \in \mathscr{U}'$, none of the entries of $Bc$ is equal to zero. Equivalently, for every $c=[c_1,\dots,c_{mn}] \in \mathscr{U}'$, the polynomial $h_{d+1}(Z) = c_1 \hat{p}_1(Z) + \cdots + c_{mn} \hat{p}_{mn}(Z)$ has the property that $h_{d+1}(Y_i) \neq 0$ for every $i \in [\ell]$.
  \end{proof}
  
 \subsection{Proof of Theorem \ref{thm:pi-ladder}} \label{subsection:proof-pi-ladder}
 
 With $Z=(z_{ij})$ the $m \times n$ matrix of variables as above, we let $Z_{\Lc} = \{z_{ij}: \, (i,j) \in \Lc\}$. We also let $\Rs_\Lc = \Re[Z_{\Lc}]$ be the polynomial ring with variables $Z_\Lc$. We let $\Is_\Lc$ be the ideal of $\Rs_\Lc$ generated by the polynomials $\det(Z_{\I,\J})$ for all squares $\I \times \J \subset \Lc$ of side length $r+1$. The basic ingredient is a generalization of Lemma \ref{lem:GB-det}: 
 
 \begin{lem}(\cite{narasimhan1986irreducibility}) \label{lem:GB-ladder}
 Under the diagonal order $\succ$ of Definition \ref{dfn:diagonal-order}, the set of all $(r+1)\times (r+1)$ determinants of $Z$ which are supported in $\Lc$ form a Gr\"obner basis of $\Is_\Lc$. 
 \end{lem}
 
 \noindent With Lemma \ref{lem:GB-ladder} at hand, the proof follows from very similar arguments as the proof of Theorem \ref{thm:arbitrary-pi}. 
 
 \subsection{Proof of Theorem \ref{thm:ladder-system}}
 
 In addition to the notation of \S \ref{subsection:proof-pi-ladder}, we let $\Ls_\Lc$ the ideal generated by the polynomials $\hat{p}_{\Lc,\nu}$ for all $\nu \in [\#\Lc]$. With this notation, the proof follows from similar arguments \textemdash{which in fact generalize the proof of Theorem \ref{thm:system}\textemdash} once we note: 
 
 \begin{lem} (\cite{gorla2007mixed,herzog1992grobner})
 $\height(\Is_\Lc) = \#\Lc - \#\Lc'$, where $\Lc'$ is the ladder of Proposition \ref{prp:Gorla}.
 \end{lem}
 
 \begin{lem} (\cite{narasimhan1986irreducibility})
 The ideal $\Is_\Lc$ is prime. 
 \end{lem}
 
 \begin{lem}(\cite{herzog1992grobner,gorla2007mixed})
 The quotient ring $\Rs_\Lc / \Is_\Lc$ is Cohen-Macaulay.
 \end{lem}
 
\appendices

\setcounter{prp}{0}
\setcounter{ex}{0}
\setcounter{thm}{0}
\setcounter{dfn}{0}
\setcounter{lem}{0}
\renewcommand{\theprp}{\Alph{section}\arabic{prp}}
\renewcommand{\theex}{\Alph{section}\arabic{ex}}
\renewcommand{\thelem}{\Alph{section}\arabic{lem}}
\renewcommand{\thethm}{\Alph{section}\arabic{thm}}
\renewcommand{\thedfn}{\Alph{section}\arabic{dfn}}

\section{Dimension Theory of Commutative Rings} \label{appendix:Ring-Theory}
We cover some commutative ring theory (\cite{dummit1991abstract,matsumura1989commutative,eisenbud2013commutative,bruns2006determinantal,bruns1998cohen}), as needed for our proofs.

We will assume that the reader is familiar with the definition of a (commutative) ring (\cite{dummit1991abstract}). The prototypical example of a ring in this paper is the set $\Re[\underline{y}] = \Re[y_1,\dots,y_n]$ of all polynomials in the variables $\underline{y}=y_1,\dots,y_n$ with real coefficients.

So let $\Rs$ be a ring. The most important algebraic structure inside $\Rs$ is that of an ideal. 

\begin{dfn}[ideal] \label{dfn:ideal}
 A subset $\Is$ of $\Rs$ is called an ideal, if it is closed under subtraction, and also under multiplication by any element of $\mathscr{R}$.
 \end{dfn}
 
 \noindent It is easy to construct ideals:
 
 \begin{dfn}[ideal generated by elements] \label{dfn:ideal-generated-by}
 Let $r_1,\dots,r_s$ be elements of $\Rs$. The set of all elements of the form $c_1r_1 + \cdots + c_s r_s \in \Rs$, where the $c_i$'s take all values in $\mathscr{R}$, is an ideal called the ideal generated by $r_1,\dots,r_s$, denoted as $(r_1,\dots,r_s)$. The $r_i$'s are called \emph{generators} of the ideal $\mathscr{I}$.
 \end{dfn}
 
 \noindent We will assume that $\Rs$ is \emph{Noetherian}: 
 
 \begin{dfn}[Noetherian ring]
 $\Rs$ is Noetherian if every ideal of $\Rs$ has a finite number of generators. 
 \end{dfn}
 
 \noindent From two ideals we can get a new ideal:
 
 \begin{dfn}[sum of ideals] \label{dfn:sum-of-ideals}
 Let $\Is$ and $\Ls$ be two ideals of $\Rs$. Their sum $\Is + \Ls$, is the ideal consisting of all elements $r + p$, with $r \in \Is, \, p \in \Ls$.  
 \end{dfn}
 
 \noindent The most important type of an ideal is a prime ideal:
 
 \begin{dfn}[prime ideal]
 A prime ideal of $\Rs$ is an ideal $\mathscr{P}$, such that whenever $r p \in \Ps$, then either $r \in \Ps$ or $p \in \Ps$.
 \end{dfn}
 
 \begin{dfn}[height of an ideal]\label{dfn:height}
 The height of a prime ideal $\Ps$ is the supremum of the length $n$ of all chains of prime ideals of the form $$ \Ps_0 \subsetneq \Ps_1 \subsetneq \Ps_2 \subsetneq \cdots \subsetneq \Ps_n=\Ps.$$ The height of an ideal $\Is$, denoted $\height (\Is)$, is the infimum of the heights of all prime ideals that contain $\Is$. 
 \end{dfn}
 
 \noindent The height is used to define the dimension of $\Rs$:
 
 \begin{dfn}[dimension of a ring]\label{dfn:ring-dimension}
 The dimension of $\Rs$, denoted $\dim (\Rs)$, is the supremum among the heights of all of its prime ideals.
 \end{dfn}
 
 \noindent Given an ideal, we can get a new ring:
 
 \begin{dfn}[quotient ring]\label{dfn:quotient-ring}
 Let $\Is$ be an ideal of $\Rs$. Then $\Rs / \Is$ denotes the ring of equivalence classes of elements of $\Rs$ under the equivalence relation $r \sim p$ if and only if $r-p \in \Is$. If $r \in \Rs$, we denote its equivalance class as $\bar{r} \in \Rs / \Is$. If $\Ls$ is another ideal of $\Rs$, the set of all $\bar{r} \in \Rs / \Is$ for $r \in \Ls$ forms an ideal of $\Rs/\Is$, denoted by $\bar{\Ls}$.
 \end{dfn} 
 
 \noindent There are some simple relationships between ideals of a ring and ideals of its quotient:
 
 \begin{prp}\label{prp:ideals-quotient}
 Let $\Is$ be an ideal of $\Rs$ and set $\bar{\Rs} = \Rs/\Is$. Consider the map $\varphi: \Rs \rightarrow \bar{\Rs}$ that takes $r$ to its equivalence class $\varphi(r)=\bar{r}$. We have that i) if $\Ls$ is an ideal of $\Rs$ then $\varphi(\Ls)$ is an ideal of $\bar{\Rs}$, which we denote by $\bar{\Ls}$, ii) every ideal of $\bar{\Rs}$ is of the form $\bar{\Ls}$, iii) if $\mathfrak{P}$ is a prime ideal of $\bar{\Rs}$, then $\varphi^{-1}(\mathfrak{P})$ is a prime ideal of $\Rs$ that contains $\Is$.
 \end{prp}
 
 \noindent The next statement, treating the quotient ring of a quotient ring, is known in abstract algebra as \emph{the second isomorphism theorem}:
 
 \begin{prp}\label{prp:second-isomorphism-theorem}
 Let $\Is$ and $\Ls$ be ideals of $\Rs$. Set $\bar{\Ls}$ the ideal of $\bar{\Rs}=\Rs/\Is$ induced by $\Ls$. Then there is an isomorphism of rings $$ \bar{\Rs} / \bar{\Ls} \cong \Rs / \Is + \Ls.$$
 \end{prp}
 
 \noindent The dimension of $\Rs$ is controlled by a special family of prime ideals: 
 
 \begin{prp}[minimal primes] \label{prp:minimal-primes}
 There is a finite set of prime ideals of $\Rs$, denoted $\Min(\Rs)$ and called the minimal primes of $\Rs$, such that every $\Ps \in \Min(\Rs)$ does not contain any other prime ideal besides itself. Moreover, $$\dim(\Rs) = \max_{\Ps \in \Min(\Rs)} \dim(\Rs/\Ps).$$
 \end{prp}
 
 \begin{dfn}[equidimensional ring] \label{dfn:equidimensional-ring}
 $\Rs$ is called an equidimensional ring, if $\dim (\Rs/\Ps) = \dim (\Rs)$ for every $\Ps \in \Min(\Rs)$. \end{dfn}
 
 \noindent An element $r \in \Rs$ is either \emph{regular} or a \emph{zero-divisor}:  
 
 \begin{dfn}[zero-divisor, regular element]
 An element $r \in \Rs$ is called a zero-divisor, if $r p = 0$ for some non-zero $p \in \Rs$. If $r$ is not a zero-divisor, then it is called an $\Rs$-regular element.
 \end{dfn}
 
 \noindent The zero-divisors are controlled by another special family of prime ideals:
 
 \begin{prp}[associated primes] \label{prp:associated-primes}
 There is a finite set of prime ideals of $\Rs$, denoted $\Ass(\Rs)$ and called the associated primes of $\Rs$, such that the set $\cup_{\Ps \in \Ass(\Rs)} \Ps$ is the set of all zero-divisors of $\Rs$.
 \end{prp}
 
 Following \cite{hartshorne1977algebraic}, Definition A1 in \cite{tsakiris2020algebraic} gave a geometric notion for the dimension of an algebraic variety. In Proposition B1 of \cite{tsakiris2020algebraic} an algebraic characterization was given. The two descriptions \textemdash{geometric and algebraic\textemdash} coincide whenever the underlying field is $\mathbb{C}$ or when one interprets the variety as a \emph{scheme} \cite{hartshorne1977algebraic}. In this paper neither applies: we work over $\Re$ and we consider algebraic varieties in the classical \textemdash{non-scheme\textemdash} sense. It turns out the correct definition in this case is the algebraic one: 
   
 \begin{dfn}[dimension of algebraic variety] \label{dfn:dimension-variety}
 Let $\Rs=\Re[\underline{y}]$ and let $\Is$ be an ideal. Consider the variety $\V(\Is)$ defined by $\Is$, that is the set of all $\xi \in \Re^n$ such that $p(\xi)=0$ for every $p \in \Is$. The dimension of $\V(\Is)$ is defined to be $\dim (\Rs/ \Is)$. 
 \end{dfn}
 
 \noindent A key notion in the dimension theory of rings is:
  
 \begin{dfn}[regular sequence]
 An ordered set of elements $r_1, \dots, r_k \in \Rs$ is called an $\Rs$-regular sequence of length $k$, if $r_1$ is $\Rs$-regular, and for every $i=2,\dots,k$ the element $\bar{r}_i \in \Rs/(r_1,\dots,r_{i-1})$ is $\Rs / (r_1,\dots,r_{i-1})$-regular. 
 \end{dfn}
 
 \noindent For well-behaved classes of rings there is a nice formula \textemdash{the expected one\textemdash} relating the dimension of $\Rs/ \Is$ with the height of $\Is$. Here we only need it for quotients of polynomial rings:
 
 \begin{prp} \label{prp:dimension-formula}
 Let $\Rs=\Re[\underline{y}]$ and $\Ps$ be a prime ideal. Set $\bar{\Rs} = \Rs / \Ps$ and $\bar{\Ls}$ any ideal of $\bar{\Rs}$. Then $$ \height (\bar{\Ls}) + \dim (\bar{\Rs}/\bar{\Ls}) = \dim (\bar{\Rs}).$$
 \end{prp}
 
 \noindent Another useful notion attached to an ideal is: 
 
 \begin{dfn}[grade of an ideal]
 The grade of an ideal $\Is$, denoted $\grade (\Is)$, is the length of the largest $\Rs$-regular sequence contained in $\Is$.
 \end{dfn}
 
 \noindent The grade has the expected behavior, when we pass to the quotient by regular elements:
 
 \begin{prp}\label{prp:grade-regular}
 Let $\Is$ be an ideal of $\Rs$ with $\grade(\Is)>0$. Let $r \in \Is$ be a $\Rs$-regular element. Set $\bar{\Rs} = \Rs / (r)$ and let $\bar{\Is}$ be the ideal of $\bar{\Rs}$ induced by $\Is$. Then $\grade(\bar{\Is}) = \grade(\Is)-1.$
 \end{prp}
 
 \noindent When we pass to the quotient by a regular element, the dimension also drops by $1$:
 
 \begin{prp} \label{prp:equidimensional-drop-1}
 Let $\Rs = \Re[\underline{y}]$ and $\Is$ an ideal of $\Rs$. Suppose that $\bar{\Rs}=\Rs/\Is$ is equidimensional. Let $\bar{r} \in \bar{\Rs}$ be an $\bar{\Rs}$-regular element, such that $\Is + (r) \neq \Rs$. Then $\bar{\bar{\Rs}} = \bar{\Rs} / (\bar{r})$ is also equidimensional and $\dim (\bar{\bar{\Rs}}) = \dim (\bar{\Rs})-1$. 
 \end{prp} 
 
 \noindent The height always bounds from above the grade:
 \begin{prp} \label{prp:grade-height}
 Let $\Is$ be an ideal of $\Rs$. Then $$  \grade(\Is) \le \height(\Is).$$
 \end{prp}
 
 \noindent The following is a well-behaved class of rings:
 
 \begin{dfn}[Cohen-Macaulay ring]\label{dfn:Cohen-Macaulay ring}
 $\Rs$ is called a Cohen-Macaulay ring if for every ideal $\Is$ we have $$ \grade(\Is) = \height(\Is).$$
 \end{dfn}

\setcounter{prp}{0}
\setcounter{ex}{0}
\setcounter{thm}{0}
\setcounter{dfn}{0}
\setcounter{lem}{0}
\renewcommand{\theprp}{\Alph{section}\arabic{prp}}
\renewcommand{\theex}{\Alph{section}\arabic{ex}}
\renewcommand{\thelem}{\Alph{section}\arabic{lem}}
\renewcommand{\thethm}{\Alph{section}\arabic{thm}}
\renewcommand{\thedfn}{\Alph{section}\arabic{dfn}}

\section{Gr\"obner Bases} \label{appendix:GB}

The books \cite{cox2013ideals} and \cite{michalek2021invitation} offer excellent and accessible introductions to Gr\"obner bases. We also mention the beautiful treatment of Chapter 1 in \cite{Aldo-book}. Here, we give the necessary background to make our proofs as self-contained as possible. 

Let $\mathscr{R}=\Re[y] = \Re[y_1,\dots,y_n]$ be the polynomial ring of all polynomials in variables $y_1,\dots,y_n$ with real coefficients. Let $p_1,\dots,p_s \in \Re[y]$ be polynomials and let $\mathscr{I} = (p_1,\dots,p_s)$ be the ideal they generate. A Gr\"obner basis of $\mathscr{I}$ is a particular set of generators of $\mathscr{I}$, which possesses several favorable properties. These properties pertain to the notion of division of multivariate polynomials, which is not as simple as division of polynomials in one variable. 
A key concept that arises in this context is:

\begin{dfn}[\emph{monomial order}] \label{dfn:monomial-order}
Write a monomial $y_1^{\alpha_1} \cdots y_n^{\alpha_n}$ of $\mathscr{R}$ as $y^\alpha$, with $\alpha = (\alpha_1,\dots,\alpha_n) \in \mathbb{N}^n$. A monomial order of $\mathscr{R}$ is a total order $\succ$ on all of its monomials satisfying 1) $ y^\alpha \succ 1$ for any $\alpha$, and 2) if $y^\beta \succ y^\alpha$ then $ y^{\beta} y^{\gamma} \succ y^{\alpha} y^{\gamma}$ for any $\gamma \in \mathbb{N}^n$. 
\end{dfn} 

\noindent An example of a monomial order is the \emph{lexicographic order}. There are many lexicographic orders on $\mathscr{R}$ depending on how we order the $y_i$'s.

\begin{dfn}[lexicographic order] \label{dfn:lexicographic-order}
 The lexicographic order on $\mathscr{R}$ with $y_1 \succ y_2 \succ \cdots \succ y_n$, is defined by $y^\beta \succ y^\alpha$, if the first non-zero entry of $\beta - \alpha$ is positive. 
 \end{dfn}
 
 \noindent Polynomials are built up from monomials:
 
 \begin{dfn}[support of polynomial] \label{dfn:support}
 Let $p \in \Rs$ be a polynomial. We can write it as $p = \sum_{u} c_\alpha y^\alpha$, where the summation is over all distinct monomials $y^\alpha$ of $\Rs$. The support of $p$ is defined as the set of monomials for which $c_\alpha \neq 0$.
 \end{dfn}
 
A monomial order is a total order, that is any two monomials are comparable. The largest monomial in the support of a polynomial plays a special role.

\begin{dfn}[initial monomial] \label{dfn:initial-monomial}
Let $\succ$ be a monomial order on $\Rs$ and $p $ a non-zero polynomial. The largest monomial that appears in the support of $p$ is called the initial monomial of $p$ and denoted $\operatorname{in}_{\succ}(p)$. 
\end{dfn}

\begin{dfn}[\emph{initial ideal}] \label{dfn:initial-ideal}
Let $\mathscr{I}$ be an ideal of $\mathscr{R}$ and $\succ$ a monomial order. The initial ideal of $\mathscr{I}$, denoted $\operatorname{in}_{\succ}(\mathscr{I})$, is the ideal of $\mathscr{R}$ generated by all initial monomials $\operatorname{in}_{\succ}(p)$, for all $p \in \mathscr{I}$. 
\end{dfn}

\noindent We can now define the notion of a Gr\"obner basis:

\begin{dfn}[\emph{Gr\"obner basis}]\label{dfn:GB}
Let $\mathscr{I}$ be an ideal of $\mathscr{R}$ and $\succ$ a monomial order. A Gr\"obner basis of $\mathscr{I}$ is a set of generators $g_1,\dots,g_k$ of $\mathscr{I}$, such that their initial monomials generate the initial ideal of $\mathscr{I}$, that is $\operatorname{in}_{\succ}(\mathscr{I}) = (\operatorname{in}_{\succ}(g_1),\dots,\operatorname{in}_{\succ}(g_k))$.
\end{dfn}

\noindent We will also need the following simple fact:

\begin{lem} \label{lem:monomial-ideal}
Let $\Is$ be an ideal generated by monomials of $\Rs$, say $y^{\alpha_1}, \dots, y^{\alpha_k}$, where $\alpha_i = (\alpha_{i1},\dots,\alpha_{in}) \in \mathbb{N}^n$. Then a monomial $y^{\beta}$ lies in $\Is$ if and only if it is divisible by one of the $y^{\alpha_i}$'s.
\end{lem}

\section*{Acknowledgement} 
The author is grateful to Aldo Conca for indicating Lemma \ref{lem:Conca} with the elegant \emph{Vandermode trick} used in its proof, to Yang Yang (ORCID 0000-0003-1200-268X) for indicating examples of data in neuroscience and biology which possess a ladder-like structure, and to Liangzu Peng for useful comments on the manuscript. 

\bibliographystyle{IEEEtran}
\bibliography{LMRP-arXiv-July22}

\begin{thebibliography}{10}
\providecommand{\url}[1]{#1}
\csname url@samestyle\endcsname
\providecommand{\newblock}{\relax}
\providecommand{\bibinfo}[2]{#2}
\providecommand{\BIBentrySTDinterwordspacing}{\spaceskip=0pt\relax}
\providecommand{\BIBentryALTinterwordstretchfactor}{4}
\providecommand{\BIBentryALTinterwordspacing}{\spaceskip=\fontdimen2\font plus
\BIBentryALTinterwordstretchfactor\fontdimen3\font minus
  \fontdimen4\font\relax}
\providecommand{\BIBforeignlanguage}[2]{{%
\expandafter\ifx\csname l@#1\endcsname\relax
\typeout{** WARNING: IEEEtran.bst: No hyphenation pattern has been}%
\typeout{** loaded for the language `#1'. Using the pattern for}%
\typeout{** the default language instead.}%
\else
\language=\csname l@#1\endcsname
\fi
#2}}
\providecommand{\BIBdecl}{\relax}
\BIBdecl

\bibitem{Unnikrishnan-Allerton2015}
J.~{Unnikrishnan}, S.~{Haghighatshoar}, and M.~{Vetterli}, ``Unlabeled sensing:
  Solving a linear system with unordered measurements,'' in \emph{Annual
  Allerton Conference on Communication, Control, and Computing}, 2015, pp.
  786--793.

\bibitem{Unnikrishnan-TIT18}
J.~Unnikrishnan, S.~Haghighatshoar, and M.~Vetterli, ``Unlabeled sensing with
  random linear measurements,'' \emph{IEEE Transactions on Information Theory},
  vol.~64, no.~5, pp. 3237--3253, May 2018.

\bibitem{tsakiris2018eigenspace}
M.~C. Tsakiris, ``Determinantal conditions for homomorphic sensing,''
  \emph{preprint arXiv:1812.07966v1-v6}, 2018-2020.

\bibitem{Tsakiris-ICML2019}
M.~C. Tsakiris and L.~Peng, ``Homomorphic sensing,'' in \emph{International
  Conference on Machine Learning}, 2019, pp. 6335--6344.

\bibitem{Peng-ACHA-21}
L.~Peng and M.~C. Tsakiris, ``Homomorphic sensing of subspace arrangements,''
  \emph{Applied and Computational Harmonic Analysis}, vol.~55, pp. 466--485,
  2021.

\bibitem{Hsu-NIPS17}
D.~Hsu, K.~Shi, and X.~Sun, ``Linear regression without correspondence,'' in
  \emph{Advances in Neural Information Processing Systems}, vol.~30, 2017.

\bibitem{Pananjady-TIT18}
A.~Pananjady, M.~J. Wainwright, and T.~A. Courtade, ``Linear regression with
  shuffled data: Statistical and computational limits of permutation
  recovery,'' \emph{IEEE Transactions on Information Theory}, vol.~64, no.~5,
  pp. 3286--3300, 2018.

\bibitem{Slawski-JoS19}
M.~Slawski and E.~Ben-David, ``Linear regression with sparsely permuted data,''
  \emph{Electronic Journal of Statistics}, vol.~13, no.~1, pp. 1--36, 2019.

\bibitem{zhang2020optimal}
H.~Zhang and P.~Li, ``Optimal estimator for unlabeled linear regression,'' in
  \emph{International Conference on Machine Learning}.\hskip 1em plus 0.5em
  minus 0.4em\relax PMLR, 2020, pp. 11\,153--11\,162.

\bibitem{tsakiris2020algebraic}
M.~C. Tsakiris, L.~Peng, A.~Conca, L.~Kneip, Y.~Shi, and H.~Choi, ``An
  algebraic-geometric approach for linear regression without correspondences,''
  \emph{IEEE Transactions on Information Theory}, 2020.

\bibitem{jeong2020recovering}
M.~Jeong, A.~Dytso, M.~Cardonea, and V.~H. Poor, ``Recovering data permutations
  from noisy observations: The linear regime,'' \emph{IEEE Journal on Selected
  Areas in Information Theory}, vol.~1, no.~3, pp. 854--869, 2020.

\bibitem{Slawski-JMLR2020}
M.~Slawski, E.~Ben-David, and P.~Li, ``Two-stage approach to multivariate
  linear regression with sparsely mismatched data,'' \emph{Journal of Machine
  Learning Research}, vol.~21, no. 204, pp. 1--42, 2020.

\bibitem{Slawski-JCGS2021}
M.~Slawski, G.~Diao, and E.~Ben-David, ``A pseudo-likelihood approach to linear
  regression with partially shuffled data,'' \emph{Journal of Computational and
  Graphical Statistics}, vol.~0, no.~0, pp. 1--31, 2021.

\bibitem{li2021generalized}
F.~Li, K.~Fujiwara, F.~Okura, and Y.~Matsushita, ``Generalized shuffled linear
  regression,'' in \emph{Proceedings of the IEEE/CVF International Conference
  on Computer Vision}, 2021, pp. 6474--6483.

\bibitem{Abid-Allerton2018}
A.~{Abid} and J.~{Zou}, ``A stochastic expectation-maximization approach to
  shuffled linear regression,'' in \emph{Annual Allerton Conference on
  Communication, Control, and Computing}, 2018, pp. 470--477.

\bibitem{xie2021hypergradient}
Y.~Xie, Y.~Mao, S.~Zuo, H.~Xu, X.~Ye, T.~Zhao, and H.~Zha, ``A hypergradient
  approach to robust regression without correspondence,'' in
  \emph{International Conference on Learning Representations}, 2021.

\bibitem{ma2021optimal}
R.~Ma, T.~Cai, and H.~Li, ``Optimal permutation recovery in permuted monotone
  matrix model,'' \emph{Journal of the American Statistical Association}, vol.
  116, no. 535, pp. 1358--1372, 2021.

\bibitem{nejatbakhsh2021neuron}
A.~Nejatbakhsh and E.~Varol, ``Neuron matching in c. elegans with robust
  approximate linear regression without correspondence,'' in \emph{Proceedings
  of the IEEE/CVF Winter Conference on Applications of Computer Vision}, 2021,
  pp. 2837--2846.

\bibitem{kumar2017bandlimited}
A.~Kumar, ``On bandlimited field estimation from samples recorded by a
  location-unaware mobile sensor,'' \emph{IEEE Transactions on Information
  Theory}, vol.~63, no.~4, pp. 2188--2200, 2017.

\bibitem{wang2020target}
G.~Wang, S.~Marano, J.~Zhu, and Z.~Xu, ``Target localization by unlabeled range
  measurements,'' \emph{IEEE Transactions on Signal Processing}, vol.~68, pp.
  6607--6620, 2020.

\bibitem{yao2021}
Y.~Yao, L.~Peng, and M.~C. Tsakiris, ``Unlabeled principal component
  analysis,'' \emph{Advances in Neural Information Processing Systems},
  vol.~34, pp. 30\,452--30\,464, 2021.

\bibitem{eisenbud2013commutative}
D.~Eisenbud, \emph{Commutative algebra: with a view toward algebraic
  geometry}.\hskip 1em plus 0.5em minus 0.4em\relax Springer Science \&
  Business Media, 2013, vol. 150.

\bibitem{cox2013ideals}
D.~Cox, J.~Little, and D.~OShea, \emph{{Ideals, Varieties, and Algorithms: An
  Introduction to Computational Algebraic Geometry and Commutative
  Algebra}}.\hskip 1em plus 0.5em minus 0.4em\relax Springer Science \&
  Business Media, 2013.

\bibitem{sturmfels1990grobner}
B.~Sturmfels, ``Gr{\"o}bner bases and stanley decompositions of determinantal
  rings,'' \emph{Mathematische Zeitschrift}, vol. 205, no.~1, pp. 137--144,
  1990.

\bibitem{bruns2003grobner}
W.~Bruns and A.~Conca, ``Gr{\"o}bner bases and determinantal ideals,'' in
  \emph{Commutative algebra, singularities and computer algebra}.\hskip 1em
  plus 0.5em minus 0.4em\relax Springer, 2003, pp. 9--66.

\bibitem{peters2014emergence}
A.~J. Peters, S.~X. Chen, and T.~Komiyama, ``Emergence of reproducible
  spatiotemporal activity during motor learning,'' \emph{Nature}, vol. 510, no.
  7504, pp. 263--267, 2014.

\bibitem{he2020changing}
P.~He, B.~A. Williams, D.~Trout, G.~K. Marinov, H.~Amrhein, L.~Berghella
  \emph{et~al.}, ``The changing mouse embryo transcriptome at whole tissue and
  single-cell resolution,'' \emph{Nature}, vol. 583, no. 7818, pp. 760--767,
  2020.

\bibitem{narasimhan1986irreducibility}
H.~Narasimhan, ``The irreducibility of ladder determinantal varieties,''
  \emph{Journal of Algebra}, vol. 102, no.~1, pp. 162--185, 1986.

\bibitem{conca1995ladder}
A.~Conca, ``Ladder determinantal rings,'' \emph{Journal of Pure and Applied
  Algebra}, vol.~98, no.~2, pp. 119--134, 1995.

\bibitem{conca1996gorenstein}
------, ``Gorenstein ladder determinantal rings,'' \emph{Journal of the London
  Mathematical Society}, vol.~54, no.~3, pp. 453--474, 1996.

\bibitem{conca1997ladder}
A.~Conca and J.~Herzog, ``Ladder determinantal rings have rational
  singularities,'' \emph{advances in mathematics}, vol. 132, no.~1, pp.
  120--147, 1997.

\bibitem{gorla2007mixed}
E.~Gorla, ``Mixed ladder determinantal varieties from two-sided ladders,''
  \emph{Journal of Pure and Applied Algebra}, vol. 211, no.~2, pp. 433--444,
  2007.

\bibitem{tsakiris2020low}
M.~C. Tsakiris, ``Low-rank matrix completion theory via plucker coordinates,''
  \emph{arXiv preprint arXiv:2004.12430v5}, 2021.

\bibitem{michalek2021invitation}
M.~Michalek and B.~Sturmfels, \emph{Invitation to nonlinear algebra}.\hskip 1em
  plus 0.5em minus 0.4em\relax American Mathematical Soc., 2021, vol. 211.

\bibitem{melanova2022recovery}
H.~Mel{\'a}nov{\'a}, B.~Sturmfels, and R.~Winter, ``Recovery from power sums,''
  \emph{Experimental Mathematics}, pp. 1--10, 2022.

\bibitem{song2018permuted}
X.~Song, H.~Choi, and Y.~Shi, ``Permuted linear model for header-free
  communication via symmetric polynomials,'' in \emph{2018 IEEE International
  Symposium on Information Theory (ISIT)}.\hskip 1em plus 0.5em minus
  0.4em\relax IEEE, 2018, pp. 661--665.

\bibitem{gonciulea2000mixed}
N.~Gonciulea and C.~Miller, ``Mixed ladder determinantal varieties,''
  \emph{Journal of Algebra}, vol. 231, no.~1, pp. 104--137, 2000.

\bibitem{bruns2006determinantal}
W.~Bruns and U.~Vetter, \emph{Determinantal rings}.\hskip 1em plus 0.5em minus
  0.4em\relax Springer, 1988, vol. 1327.

\bibitem{herzog1992grobner}
J.~Herzog and N.~Trung, ``Gr{\"o}bner bases and multiplicity of determinantal
  and pfaffian ideals,'' \emph{Advances in Mathematics}, vol.~96, no.~1, pp.
  1--37, 1992.

\bibitem{matsumura1989commutative}
H.~Matsumura, \emph{{Commutative Ring Theory}}.\hskip 1em plus 0.5em minus
  0.4em\relax Cambridge University Press, 1989, no.~8.

\bibitem{caniglia1990ideals}
L.~Caniglia, J.~A. Guccione, and J.~J. Guccione, ``Ideals of generic minors,''
  \emph{Communications in Algebra}, vol.~18, no.~8, pp. 2633--2640, 1990.

\bibitem{ma1994minors}
Y.~Ma, ``On the minors defined by a generic matrix,'' \emph{Journal of symbolic
  computation}, vol.~18, no.~6, pp. 503--518, 1994.

\bibitem{sturmfels2006combinatorial}
B.~Sturmfels and S.~Sullivant, ``Combinatorial secant varieties,'' \emph{Pure
  and Applied Mathematics Quarterly}, vol.~2, no.~3, pp. 867--891, 2006.

\bibitem{mount1967remark}
K.~R. Mount, ``A remark on determinantal loci,'' \emph{Journal of the London
  Mathematical Society}, vol.~1, no.~1, pp. 595--598, 1967.

\bibitem{eagon1969examples}
J.~A. Eagon, ``Examples of cohen-macauley rings which are not gorenstein,''
  \emph{Mathematische Zeitschrift}, vol. 109, no.~2, pp. 109--111, 1969.

\bibitem{northcott1963some}
D.~G. Northcott, ``Some remarks on the theory of ideals defined by matrices,''
  \emph{The Quarterly Journal of Mathematics}, vol.~14, no.~1, pp. 193--204,
  1963.

\bibitem{hochster1971cohen}
M.~Hochster and J.~Eagon, ``{Cohen-Macaulay rings, invariant theory, and the
  generic perfection of determinantal loci},'' \emph{American Journal of
  Mathematics}, vol.~93, no.~4, pp. 1020--1058, 1971.

\bibitem{conca2009regular}
A.~Conca, C.~Krattenthaler, and J.~Watanabe, ``Regular sequences of symmetric
  polynomials,'' \emph{Rendiconti del Seminario Matematico della Universit{\`a}
  di Padova}, vol. 121, pp. 179--199, 2009.

\bibitem{Aldo-book}
W.~Bruns, A.~Conca, C.~Raicu, and M.~Varbaro, \emph{Determinants, Gr\"obner
  bases and Cohomology}, Springer, 2022 (in print).

\bibitem{tsakiris2017filtrated}
M.~C. Tsakiris and R.~Vidal, ``Filtrated algebraic subspace clustering,''
  \emph{SIAM Journal on Imaging Sciences}, vol.~10, no.~1, pp. 372--415, 2017.

\bibitem{dummit1991abstract}
D.~S. Dummit and R.~M. Foote, \emph{{Abstract Algebra}}.\hskip 1em plus 0.5em
  minus 0.4em\relax Prentice Hall, 1991, vol. 1999.

\bibitem{bruns1998cohen}
W.~W.~Bruns and H.~J. Herzog, \emph{{Cohen-Macaulay Rings}}.\hskip 1em plus
  0.5em minus 0.4em\relax Cambridge university press, 1998, no.~39.

\bibitem{hartshorne1977algebraic}
R.~Hartshorne, \emph{Algebraic Geometry}.\hskip 1em plus 0.5em minus
  0.4em\relax Springer-Verlag, 1977, vol.~79.

\end{thebibliography}

\end{document}